%% file: paper.tex
\documentclass[conference]{IEEEtran}

\usepackage{amsmath}
\usepackage{listings}
\usepackage{tikz}
\usepackage{amsthm}
\usepackage{lmodern}
\usepackage[utf8]{inputenc}
\usepackage[T1]{fontenc}
\usepackage{amssymb}

\usepackage{graphicx}
\usepackage{color}
\usepackage{hyperref}

\usepackage{mathpartir}

\usepackage{xspace}
\usepackage{fixltx2e}
\usepackage{framed}
\usepackage{relsize}

\newtheorem{mydef}{Definition}
\newtheorem{myLemma}{Lemma}
\newtheorem{myThm}{Theorem}

\newtheorem{mycor}{Corollary}

\newcommand{\expr}{\mathit{e}} 
\newcommand{\comm}{\mathit{c}}

\newcommand{\pc}{\mathit{pc}} 
\newcommand{\lab}{\mathcal{A}}

\newcommand{\pl}{\ensuremath{{}^\star}\xspace}

\newcommand{\inference}[3][]{\text{\smaller #1}\inferrule{ #2}{#3}}

\IEEEoverridecommandlockouts

\begin{document}








\title{Generalizing Permissive-Upgrade in Dynamic Information Flow
  Analysis\thanks{This is an updated version of a paper of the same
    title published at the ACM Ninth Workshop on Programming Languages
    and Analysis for Security (PLAS), 2014
    (doi>10.1145/2637113.2637116). The update improves the
    permissiveness of the dynamic analysis of the original paper
    slightly.}}


\author{\IEEEauthorblockN{Abhishek Bichhawat\IEEEauthorrefmark{1},
Vineet Rajani\IEEEauthorrefmark{2},
Deepak Garg\IEEEauthorrefmark{2},
Christian Hammer\IEEEauthorrefmark{1}}
\IEEEauthorblockA{\IEEEauthorrefmark{1}Saarland University, \{bichhawat,hammer\}@cs.uni-saarland.de}
\IEEEauthorblockA{\IEEEauthorrefmark{2}MPI-SWS, \{vrajani,dg\}@mpi-sws.org}}

\maketitle

\input{abstract}

\IEEEpeerreviewmaketitle

\input{introduction}
\input{background}
\input{approach}

\input{related}

\input{conclusion}

\bibliographystyle{IEEEtranS.bst}
\bibliography{ifc}

\input{appendix}





\end{document}

%% file: abstract.tex
\begin{abstract}
  Preventing implicit information flows by dynamic program analysis
  requires coarse approximations that result in false positives,
  because a dynamic monitor sees only the executed trace of the
  program. One widely deployed method is the no-sensitive-upgrade
  check, which terminates a program whenever a variable's taint is
  upgraded (made more sensitive) due to a control dependence on
  tainted data. Although sound, this method is restrictive, e.g., it
  terminates the program even if the upgraded variable is never used
  subsequently. To counter this, Austin and Flanagan introduced the
  permissive-upgrade check, which allows a variable upgrade due to
  control dependence, but marks the variable ``partially-leaked''. The
  program is stopped later if it tries to use the partially-leaked
  variable. Permissive-upgrade handles the dead-variable assignment
  problem and remains sound. However, Austin and Flanagan develop
  permissive-upgrade only for a two-point (low-high) security lattice
  and indicate a generalization to pointwise products of such
  lattices. In this paper, we develop a non-trivial and non-obvious
  generalization of permissive-upgrade to arbitrary lattices. The key
  difficulty lies in finding a suitable notion of partial leaks that
  is both sound and permissive and in developing a suitable definition
  of memory equivalence that allows an inductive proof of soundness.
\end{abstract}




%% file: introduction.tex

\section{Introduction}
\label{sec:intro}

Information flow control (IFC) is often used to enforce
confidentiality and integrity of data. In a language-based setting,
IFC may be enforced statically~\cite{denning77,Volpano96,hunt2006:types,
  popl99,guarnieri11ISSTA,stagedIFC}, dynamically~\cite{plas09,
  plas10, csf12, Askarov09,Sabelfeld10,SME}, or in hybrid
ways~\cite{POST14,Nentwich07,Gurvan06,Gurvan07}. We are particularly
interested in dynamic IFC and, more specifically, dynamic IFC for
JavaScript, which has features like runtime code construction and
runtime modification of scope chains that make static analysis
difficult.

Dynamic IFC usually works by tracking taints or labels on individual
program values in the language runtime. A label represents a mandatory
access policy on the value. For example, the label $L$ (low
confidentiality) conventionally means that data may be read by an
(unspecified but fixed) adversary and $H$ (high confidentiality) means
the opposite. More generally, labels may be drawn from any lattice of
policies, with higher labels representing more restrictive policies. A
value $v$ labeled $\lab$ is written $v^\lab$. IFC analysis
\emph{propagates} labels as data flows during program execution. Flows
are of two kinds. \emph{Explicit} flows are induced by expression
evaluation and variable assignment. For example, if either variable
$y$ or $z$ is labeled $H$ (confidential), then the result of computing
$y+z$ will have label $H$, which makes it confidential as
well.\footnote{By ``$z$ is labeled $H$'' we actually mean ``the value
  in $z$ is labeled $H$''. This convention is used consistently.}

\emph{Implicit} flows are induced by control flow dependencies. For
example, in the program of Listing~\ref{lst1}, the value in variable
$x$ at the end of line~\ref{lineref} depends on the value in $z$ (so
the value in $x$ at the end of line~\ref{lineref} must be labeled $H$
if the value in $z$ is confidential), but $x$ is never assigned any
expression that explicitly depends on $z$. To track such implicit
flows, dynamic IFC maintains an additional taint, usually called the
program counter taint or program context taint or $pc$, which is an
upper bound on the control dependencies that lead to the current
instruction being executed. In our example, if $z$ is labeled $H$,
then at line~\ref{lineref}, $pc = H$ because of the branch in
line~\ref{linerefcond1} that depends on $z$. By tracking $pc$, dynamic
IFC can enforce that $x$ has label $H$ at the end of
line~\ref{lineref}, thus taking into account the control dependency.


However, simply tracking control flow dependencies via $pc$ is not
enough to guarantee absence of information flows when labels are
flow-sensitive, i.e., when the same variable may hold values with
different labels depending on what program paths are executed. The
program in Listing~\ref{lst1} is a classic counterexample, taken
from~\cite{plas09}. Assume that $z$ is labeled $H$ and $x$ and $y$ are
labeled $L$ initially. We compute the final value in $y$ as a function
of the value in $z$. If $z$ contains $\texttt{true}^H$, then $y$ ends
with $\texttt{true}^L$: The branch on line~\ref{linerefcond1} is not
taken, so $x$ remains $\texttt{false}^L$ at
line~\ref{linerefcond}. Hence, the branch on line~\ref{linerefcond} is
taken, but $pc = L$ at line~\ref{lineref5} and $y$ ends with
$\texttt{true}^L$. If $z$ contains $\texttt{false}^H$, then similar
reasoning shows that $y$ ends with $\texttt{false}^L$. Consequently,
in both cases $y$ ends with label $L$ and its value is exactly equal
to the value in $z$. Hence, an adversary can deduce the value of $z$
by observing $y$ at the end (which is allowed because $y$ ends with
label $L$). So, this program leaks information about $z$ despite
correct use of $pc$.

\begin{lstlisting}[float,caption=Implicit flow from $z$ to
  $y$,label=lst1][escapechar=@]
$x =$ false, $y =$ false
if (not$(z)$)@\label{linerefcond1}@
  $x =$ true@\label{lineref}@
if (not$(x)$)@\label{linerefcond}@
  $y =$ true@\label{lineref5}@
\end{lstlisting}

\begin{lstlisting}[float,label=lst1.1,caption=Impermissiveness of the NSU strategy][escapechar=@]
$x =$ false
if (not$(z)$)
  $x =$ true@\label{lineref1}@
if ($y$) f() else g()
$x =$ false@\label{lineref2}@
\end{lstlisting}

Handling such flows in dynamic IFC requires coarse approximation
because a dynamic monitor only sees program branches that are executed
and does not know what assignments may happen in alternate branches in
other executions. One such coarse approximation is the
\emph{no-sensitive-upgrade} (NSU) check proposed by
Zdancewic~\cite{zdancewic02PhD}. In the program in Listing~\ref{lst1},
we upgrade $x$'s label from $L$ to $H$ at line~\ref{lineref} in one of
the two executions above, but not the other. Subsequently, information
leaks in the other execution (where $x$'s label remains $L$) via the
branch on line~\ref{linerefcond}. The NSU check stops the leak by
preventing the assignment on line~\ref{lineref}. More generally, it
stops a program whenever a variable's label is upgraded due to a high
$pc$. This check suffices to provide \emph{termination-insensitive
  noninterference}, a standard security
property~\cite{Volpano96}. However, terminating a program
pre-emptively in this manner is quite restrictive in practice. For
example, consider the program of Listing~\ref{lst1.1}, where $z$ is
labeled $H$ and $y$ is labeled $L$. This program potentially upgrades
variable $x$ at line~\ref{lineref1} under a high $pc$, and then
executes function $\texttt{f}$ when $y$ is $\texttt{true}$ and
executes function $\texttt{g}$ otherwise. Suppose that $\texttt{f}$
does not read $x$. Then, for $y \mapsto \texttt{true}^L$, this program
leaks no information, but the NSU check would terminate this program
prematurely at line~\ref{lineref1}. (Note: $\texttt{g}$ may read $x$,
so $x$ is not a dead variable at line~\ref{lineref1}.)


To allow a dynamic IFC to accept such safe executions of programs with
variable upgrades due to high $pc$, Austin and Flanagan proposed a
less restrictive strategy called
\emph{permissive-upgrade}~\cite{plas10}. Whereas NSU stops a program
when a variable's label is upgraded due to assignment in a high $pc$,
permissive-upgrade allows the assignment, but labels the variable
\emph{partially-leaked} or $P$. The taint $P$ roughly means that the
variable's content in this execution is $H$, but it may be $L$ in
other executions. The program must be stopped later if it tries to use
or case-analyze the variable (in particular, branching on a
partially-leaked Boolean variable is stopped). Permissive-upgrade also
ensures termination-insensitive noninterference, but is strictly more
permissive than NSU. For example, permissive-upgrade stops the leaky
program of Listing~\ref{lst1} at line~\ref{linerefcond} when $z$
contains $\texttt{false}^H$, but it allows the program of
Listing~\ref{lst1.1} to execute to completion when $y$ contains
$\texttt{true}^L$.

\paragraph{Contribution of this paper}
Although permissive-upgrade is useful, its development in literature
is incomplete so far: Austin and Flanagan's original
paper~\cite{plas10}, and work building on it, develops
permissive-upgrade for \emph{only} a two-point security lattice,
containing levels $L$ and $H$ with $L \sqsubset H$, and the new label
$P$. A generalization to a pointwise product of such two-point
lattices (and, hence, a powerset lattice) was suggested by Austin and
Flanagan in the original paper, but not fully developed. As we explain
in Section~\ref{sec:existing}, this generalization works and can be
proved sound. However, that still leaves open the question of
generalizing permissive-upgrade to arbitrary lattices. It is not even
clear hitherto that this generalization exists.



In Section~\ref{sec:gen:pus}, we show by construction that a
generalization of permissive-upgrade to arbitrary lattices does indeed
exist and that it is, in fact, non-obvious. Specifically, the rule for
adding partially-leaked labels and the definition of store (memory)
equivalence needed to prove noninterference are reasonably
involved. On powerset lattices, the resulting IFC monitor is different
from the result of the product construction, and we show that our
system can be more permissive than the product construction in some
cases. By developing this generalization, our work makes
permissive-upgrade applicable to arbitrary security lattices like
other IFC techniques and, hence, constitutes a useful contribution to
IFC literature.

%% file: background.tex
\section{Language and Basic IFC Semantics}

\begin{figure}
\begin{align*}
\expr	:=~& n~\arrowvert~x~\arrowvert~\expr_1 \odot \expr_2 \\
\comm	:=~& 
 x := \expr~\arrowvert 
 ~\comm_1;\comm_2~\arrowvert \\
& \texttt{if}~\expr~\texttt{then}~\comm_1~\texttt{else}~\comm_2~\arrowvert \\
& \texttt{while}~\expr~\texttt{do}~\comm_1 \\ \\
\lab          :=~& L~\arrowvert~H\\
\pc          :=~& \lab \\
k,l,m :=~& \lab
\end{align*}
\caption{Syntax}\label{basic:syntax}
\end{figure}


\begin{figure*}
Expressions:\\
\begin{align*}
\inference[const: ]
{}
{\langle n, \sigma \rangle \Downarrow n^{\perp}}
\qquad
\inference[var: ]
{n^k := \sigma(x)}
{\langle x, \sigma \rangle \Downarrow n^k}
\end{align*}
\begin{align*}
\inference[oper: ]
{\expr = \expr' \odot \expr'' \qquad \langle \expr', \sigma \rangle \Downarrow
  n'^{k'} \qquad \langle \expr'', \sigma \rangle \Downarrow n''^{k''} \qquad n := n'
  \odot n'' \qquad k := k' \sqcup k''}
{\langle \expr, \sigma \rangle \Downarrow n^k}
\end{align*}
Statements:\\
\begin{align*}
\inference[seq: ]
{\langle \comm_1, \sigma \rangle \Downarrow_\pc \sigma'' \qquad \langle \comm_2, \sigma'' \rangle \Downarrow_\pc \sigma' }
{\langle \comm_1;\comm_2, \sigma \rangle \Downarrow_\pc \sigma'}
\qquad
\inference[if-else-t: ]
{\langle \expr, \sigma \rangle \Downarrow n^{\lab} \quad n = \texttt{true} \quad \langle
  \comm_1 , \sigma \rangle \Downarrow_{\pc\, \sqcup\, \lab} \sigma'}
{\langle \texttt{if}~\expr~\texttt{then}~\comm_1~\texttt{else}~\comm_2, \sigma \rangle \Downarrow_\pc \sigma'}
\end{align*}
\begin{align*}
\inference[if-else-f: ]
{\langle \expr, \sigma \rangle \Downarrow n^{\lab} \quad n = \texttt{false} \quad \langle
  \comm_2 , \sigma \rangle \Downarrow_{\pc\, \sqcup\, \lab} \sigma'}
{\langle \texttt{if}~\expr~\texttt{then}~\comm_1~\texttt{else}~\comm_2, \sigma \rangle \Downarrow_\pc \sigma'}
\qquad
\inference[while-f: ]
{\langle \expr, \sigma \rangle \Downarrow_\pc n^{\lab} \qquad n = \texttt{false} }
{\langle \texttt{while}~\expr~\texttt{do}~\comm_1, \sigma \rangle \Downarrow_\pc \sigma }
\end{align*}
\begin{align*}
\inference[while-t: ]
{\langle \expr, \sigma \rangle \Downarrow n^{\lab} \qquad n = \texttt{true} \qquad \langle
  \comm_1, \sigma \rangle \Downarrow_{\pc\, \sqcup\, \lab} \sigma'' \qquad \langle
  \texttt{while}~\expr~\texttt{do}~\comm_1, \sigma'' \rangle
  \Downarrow_{\pc\, \sqcup\, \lab}  \sigma' }
{\langle \texttt{while}~\expr~\texttt{do}~\comm_1, \sigma \rangle \Downarrow_\pc \sigma'}
\end{align*}
\caption{Semantics}\label{fig:basic-semantics}
\end{figure*}


Our technical development is based on a simple imperative language
shown in Figure~\ref{basic:syntax}.\footnote{Austin and Flanagan's
  work on permissive-upgrade is based on a $\lambda$-calculus with
  dynamic allocation, which is more general than this
  language~\cite{plas10}. However, our key ideas are orthogonal to the
  choice of language and generalize to the language of~\cite{plas10}
  easily. We use a simpler language to simplify non-essential
  technical details.}  The language's expressions include constants or
values ($n$), variables ($x$) and unspecified operators ($\odot$) to
combine them. The set of variables is fixed upfront. Labels ($\lab$)
are drawn from a fixed security lattice. For now, the lattice contains
only two labels $\{L, H\}$ with the ordering $L \sqsubset H$; we
generalize this later in the paper. Join ($\sqcup$) and meet
($\sqcap$) operations are defined as usual on the lattice. The program
counter label $\pc$ is an element of the lattice.

\subsection{IFC Semantics and NSU}
The rules in Figure~\ref{fig:basic-semantics} define the big-step
semantics of the language, including standard taint propagation for
IFC: the evaluation relation $\langle \expr, \sigma \rangle \Downarrow
n^k$ for expressions, and the evaluation relation $\langle \comm,
\sigma \rangle \Downarrow_{\pc} \sigma'$ for commands. Here, $\sigma$
denotes a store, a map from variables to labeled values of the form
$n^k$. For now, labels $k ::= \lab$; we generalize this later when we
introduce partially-leaked taints.

The evaluation relation for expressions evaluates an expression
$\expr$ and returns its value $n$ and label $k$. The label $k$ is the
join of labels of all variables occurring in $\expr$ (according to
$\sigma$). The relation for commands executes a command $\comm$ in the
context of a store $\sigma$, and the current program counter label
$\pc$, and yields a new store $\sigma'$.  The function
$\Gamma(\sigma(x))$ returns the label associated with the value in $x$
in store $\sigma$: If $\sigma(x) = n^k$, then $\Gamma(\sigma(x)) =
k$. We write $\perp$ for the least element of the lattice. Here, $\bot
= L$.

We explain the rules for evaluating commands. The rule for sequencing
$\comm_1; \comm_2$ evaluates the command $c_1$ under store $\sigma$
and the current $\pc$ label; this yields a new store $\sigma''$. It
then evaluates the command $c_2$ under store $\sigma''$ and the same
$\pc$ label, which yields the final store $\sigma'$.

The rules for \texttt{if-else} evaluate the branch condition $e$ to a
value $n$ with label $\lab$. Based on the value of $n$, one of the
branches $\comm_1$ and $\comm_2$ is executed under a $\pc$ obtained by
joining the current $\pc$ and the label $\lab$ of $n$. Similarly, the
rules for \texttt{while} evaluate the loop condition $e$ and execute
the loop command $c_1$ while $e$ evaluates to \texttt{true}. The $\pc$
for the loop is obtained by joining the current $\pc$ and the label
$\lab$ of the result of evaluating $e$.

Rules for assignment statements are conspicuously missing from
Figure~\ref{fig:basic-semantics} because they depend on the strategy
used to control implicit flows. In the remainder of this paper we
consider a number of such rules. To start, the rule for assignment
corresponding to the NSU check is shown in
Figure~\ref{fig:assn-nsu}. The rule checks that the label $l$ of the
assigned variable $x$ in the initial store $\sigma$ is at least as
high as $\pc$ (premise $\pc \sqsubseteq l$). If this condition is not
true, the program gets stuck. This is exactly the NSU check described
in Section~\ref{sec:intro}.


\begin{figure}
  \inference[assn-NSU: ] {l := \Gamma(\sigma(x)) \qquad \pc
    \sqsubseteq l \qquad \langle \expr, \sigma \rangle \Downarrow n^m
    } {\langle x := \expr, \sigma \rangle
    \Downarrow_\pc \sigma[x \mapsto n^{(\pc \sqcup m)}]}
\caption{Assignment rule for NSU}
\label{fig:assn-nsu}
\end{figure}

\subsection{Termination-Insensitive Noninterference}

The end-to-end security property usually established for dynamic IFC
is termination-insensitive noninterference (TINI). Noninterference
means (in a technical sense, formalized below) that two runs of the
same program starting from any two stores that are observationally
equivalent for any adversary end with two stores that are also
observationally equivalent for that adversary. For our observation
model, where the adversary sees only initial and final memories,
termination-insensitive means that we are willing to tolerate the
one-bit leak when an adversary checks whether or not the program
terminated (for programs with intermediate observable outputs,
termination-insensitivity may leak more than one
bit~\cite{askarov08:tini}). In particular, this discounted one-bit
leak accounts for termination due to failure of the NSU or
permissive-upgrade check. Technically, termination-insensitivity
amounts to considering only properly terminating runs in the
noninterference theorem.

Store equivalence is formalized as a relation $\sim_\lab$,
indexed by lattice elements $\lab$, representing the adversary.

\begin{mydef}
  Two labeled values $n_1^k$ and $n_2^m$ are $\lab$-equivalent,
  written $n_1^k \sim_\lab n_2^m$, iff either:
  \begin{enumerate}
  \item $(k = m) \sqsubseteq \lab$ and $n_1 = n_2$ or
  \item $k \not\sqsubseteq \lab$ and $m \not\sqsubseteq \lab$
  \end{enumerate}
\end{mydef}

This definition states that for an adversary at security level $\lab$,
two labeled values $n_1^k$ and $n_2^m$ are equivalent iff either
$\lab$ can access both values and $n_1$ and $n_2$ are equal, or it
cannot access either value ($k \not\sqsubseteq \lab$ and $m
\not\sqsubseteq \lab$). The additional constraint $k = m$ in clause
(1) is needed to prove noninterference by induction. Note that two
values labeled $L$ and $H$ respectively are distinguishable for the
$L$-adversary.

\begin{mydef}
  Two stores $\sigma_1$ and $\sigma_2$ are $\lab$-equivalent,
  written $\sigma_1 \sim_\lab \sigma_2$, iff for every variable $x$,
  $\sigma_1(x) \sim_\lab \sigma_2(x)$.
\end{mydef}

The following theorem states TINI for the NSU check. The theorem has
been proved for various languages in the past; we present it here for
completeness.

\begin{myThm}[TINI for NSU]
  With the assignment rule from Figure~\ref{fig:assn-nsu}, if
  $~\sigma_1 \sim_\lab \sigma_2$ and $\langle c, \sigma_1 \rangle
  \Downarrow_\pc \sigma_1' $ and $\langle c, \sigma_2
  \rangle\Downarrow_\pc \sigma_2' $, then $\sigma_1' \sim_\lab
  \sigma_2'$.
\end{myThm}
\begin{proof} Standard, see e.g.,~\cite{plas09}
\end{proof}

Although we have restricted our security lattice to two elements $L$
and $H$, the rules of Figures~\ref{fig:basic-semantics}
and~\ref{fig:assn-nsu}, the definition of equivalence above and the
theorem above (for NSU) are all general and work for arbitrary
lattices.
  


\section{Permissive-Upgrade on a Two-Point Lattice}
\label{sec:existing}


As described in Section~\ref{sec:intro}, the NSU check is restrictive
and halts many programs that do not leak information. To improve
permissiveness, the permissive-upgrade strategy was proposed as a
replacement for NSU by Austin and Flanagan~\cite{plas10}. However,
that development is limited to a two-point lattice $L \sqsubset H$ and
to pointwise products of such lattices. We present the key results
of~\cite{plas10} here (using modified notation and for our language)
and then build a generalization of permissive-upgrade to arbitrary
lattices in the next section. Readers should keep in mind that in this
section, the lattice has only two levels: $L$ (public) and $H$
(confidential).

We introduce a new label $P$ for ``partially-leaked''. We allow labels
$k,l,m$ on values to be either elements of the lattice ($L, H$) or
$P$. The $\pc$ can only be one of $L, H$ because branching on
partially-leaked values is prohibited. This is summarized by the
revised syntax of labels in Figure~\ref{pus:syntax}. The figure also
lifts the join operation $\sqcup$ to labels including $P$. Note that
joining any label with $P$ results in $P$. For brevity in definitions,
we also extend the order $\sqsubset$ to $L \sqsubset H \sqsubset
P$. However, $P$ is not a new ``top'' member of the lattice because it
receives special treatment in the semantic rules.

The intuition behind the partial-leak label $P$ is the following:

\begin{framed}
\noindent
  A variable with a value labeled $P$ may have been implicitly
  influenced by $H$-labeled values in this execution, but in other
  executions (obtainable by changing $H$-labeled values in the
  initial store), this implicit influence may not exist and, hence,
  the variable may be labeled $L$.
\end{framed}

\begin{figure}
\begin{align*}
\lab          :=~& L~\arrowvert~H\\
\pc          :=~& \lab\\
k,l,m :=~& \lab~\arrowvert~P\\ \\
k \sqcup k ~= ~& k\\
L \sqcup H ~=~ &H\\
L \sqcup P  ~=~&P \\
H\sqcup P  ~=~&P
\end{align*}
\caption{Syntax of labels including the partially-leaked label
  $P$}\label{pus:syntax}
\end{figure}

The rule for assignment with permissive-upgrade is
\begin{align*}
  \inference[assn-PUS: ] {l := \Gamma(\sigma(x)) \qquad \langle
    \expr, \sigma \rangle \Downarrow n^m} {\langle x := \expr, \sigma
    \rangle \Downarrow_\pc \sigma[x \mapsto n^{k}]}
\end{align*}
where $k$ is defined as follows:
\[ k = \left\lbrace
\begin{array}{l}
m \mbox{ if } \pc = L \\
m \sqcup H  \mbox{ if } \pc = H \mbox{ and } l = H\\
P \mbox{ otherwise}
\end{array} \right.
\]
The first two conditions in the definition of $k$ correspond to the
NSU rule (Figure~\ref{fig:assn-nsu}). The third condition applies, in
particular, when we assign to a variable whose initial label is $L$
with $\pc = H$. The NSU check would stop this assignment. With
permissive-upgrade, however, we can give the updated variable the
label $P$, consistent with the intuitive meaning of $P$. This allows
more permissiveness by allowing the assignment to proceed in all
cases. To compensate, we disallow any program (in particular, an
adversarial program) from case analyzing any value labeled
$P$. Consequently, in the rules for \texttt{if-then} and
\texttt{while} (Figure~\ref{fig:basic-semantics}), we require that the
label of the branch condition be of form $\lab$, which does not
include $P$.

The noninterference result obtained for NSU earlier can be extended to
permissive-upgrade by changing the definition of store 
equivalence. Because no program can case-analyze a $P$-labeled value,
such a value is equivalent to any other labeled value.

\begin{mydef}
  Two labeled values $n_1^k$ and $n_2^m$ are equivalent,
  written $n_1^k \sim n_2^m$, iff either:
  \begin{enumerate}
  \item $(k = m) = L$ and $n_1 = n_2$ or
  \item $k = m = H$ or
  \item $k = P$ or $m = P$
  \end{enumerate}
\end{mydef}

\begin{myThm}[TINI for permissive-upgrade with a two-point lattice]
  With the assignment rule \emph{assn-PUS}, if
  $~\sigma_1 \sim \sigma_2$ and $\langle c, \sigma_1 \rangle
  \Downarrow_\pc \sigma_1' $ and $\langle c, \sigma_2
  \rangle\Downarrow_\pc \sigma_2' $, then $\sigma_1' \sim
  \sigma_2'$.
\end{myThm}
\begin{proof} See~\cite{plas10}.
\end{proof}

Note that the above definition and proof are specific to the two-point
lattice.

\begin{figure*}
\begin{align*}
\inference[assn-1: ]
{l := \Gamma(\sigma(x)) \qquad  \langle \expr,
 \sigma \rangle \Downarrow n^m \qquad l = \lab_x \vee l = \lab_x\pl \qquad \pc
  \sqsubseteq \lab_x \qquad k := \pc \sqcup m}
{\langle x := \expr, \sigma \rangle \Downarrow_\pc \sigma[x \mapsto
  n^{k}]}
\end{align*}
\begin{align*}
\inference[assn-2: ]
{l := \Gamma(\sigma(x)) \qquad  \langle \expr,
  \sigma \rangle \Downarrow n^m \qquad l = \lab_x \vee  l = \lab_x\pl \qquad \pc
  \not\sqsubseteq \lab_x \qquad k := ((\pc \sqcup m)\sqcap \lab_x )\pl}
{\langle x := \expr, \sigma \rangle \Downarrow_\pc \sigma[x \mapsto
  n^{k}]}
\end{align*}
\caption[Caption]{Assignment rules for the generalized
  permissive-upgrade\footnotemark}\label{fig:assn-our}
\end{figure*}

\paragraph{Generalization from~\cite{plas10}}
Austin and Flanagan point out that permissive-upgrade on a two-point
lattice, as described above, can be generalized to a pointwise product
of such lattices. Specifically, let $X$ be an index set --- these
indices are called principals in~\cite{plas10}. Let a label $l$ be a
map of type $X \rightarrow \{L, H, P\}$ and let the subclass of pure
labels contain maps $\lab, \pc$ of type $X \rightarrow \{L, H\}$. The
order $\sqsubset$ and the join operation $\sqcup$ can be generalized
pointwise to these labels. Finally, the rule assn-PUS can be
generalized pointwise by replacing it with the following rule:
\begin{align*}
  \inference[assn-PUS': ] {l := \Gamma(\sigma(x)) \qquad \langle
    \expr, \sigma \rangle \Downarrow n^m} {\langle x := \expr, \sigma
    \rangle \Downarrow_\pc \sigma[x \mapsto n^{k}]}
\end{align*}
where $k$ is defined as follows:
\[ k(a) = \left\lbrace
\begin{array}{l}
m(a) \mbox{ if } \pc(a) = L \\
m(a) \sqcup H  \mbox{ if } \pc(a) = H \mbox{ and } l(a) = H\\
P \mbox{ otherwise}
\end{array} \right.
\]
It can be shown that for any semantic
derivation in this generalized system, projecting all labels to a
given principal yields a valid semantic derivation in the system with
a two-point lattice. This immediately implies noninterference for the
generalized system, where observations are limited to individual
principals.

\begin{mydef}
\label{def:eq-existing}
  Two labeled values $n_1^k$ and $n_2^m$ are $a$-equivalent,
  written $n_1^k \approx^a n_2^m$, iff either:
  \begin{enumerate}
  \item $k(a) = m(a) = L$ and $n_1 = n_2$ or
  \item $k(a) = m(a) = H$ or
  \item $k(a) = P$ or $m(a) = P$
  \end{enumerate}
\end{mydef}

\begin{myThm}[TINI for permissive-upgrade with a product lattice]
  With the assignment rule assn-PUS', if $~\sigma_1
  \approx^a \sigma_2$ and $\langle c, \sigma_1 \rangle \Downarrow_\pc
  \sigma_1' $ and $\langle c, \sigma_2 \rangle\Downarrow_\pc \sigma_2'
  $, then $\sigma_1' \approx^a \sigma_2'$.
\end{myThm}
\begin{proof} Outlined above.
\end{proof}

\paragraph{Remark} This generalization also makes sense if the 
principals are pre-ordered by a relation, say, $\leq$, with $a \leq b$
meaning that ``if $a$ has access, then $b$ must have access''. It can
be proved that the following is an \emph{invariant} on all labels $l$
that arise during program execution: $((a \leq b) \mathrel{\wedge}
(l(a) = L)) \Rightarrow l(b) = L$. Hence, the intuitive meaning of the
order $\leq$ is preserved during execution.

This generalization of the two-point lattice to an arbitrary product
of such lattices is interesting because an arbitrary powerset lattice
can be simulated using such a product. However, this still leaves open
the question of constructing a generalization of permissive-upgrade to
an arbitrary lattice. We develop such a generalization in the next
section.

%% file: approach.tex
\section{Permissive-Upgrade on Arbitrary Lattices}
\label{sec:gen:pus}

The generalization of permissive-upgrade described in this section
applies to an arbitrary security lattice. For every element $\lab$ of
the lattice, we introduce a new label $\lab\pl$ which means
``partially-leaked $\lab$'', with the following intuition.
\begin{framed}
\noindent
A variable labeled $\lab\pl$ may contain partially-leaked data, where
$\lab$ is a \emph{lower-bound} on the $\star$-free labels the variable
may have in alternate executions.
\end{framed}

\begin{figure}
\begin{align*}
\lab          :=~& L~\arrowvert~M~\arrowvert~\ldots~\arrowvert~H\\
\pc          :=~& \lab \\
k,l,m :=~& \lab~\arrowvert~\lab\pl \\ \\
\lab_1 \sqcup \lab_2\pl  :=~&(\lab_1 \sqcup \lab_2)\pl \\
\lab_1\pl  \sqcup \lab_2\pl  :=~&(\lab_1 \sqcup \lab_2)\pl
\end{align*}
\caption{Labels and label operations}\label{fig:labels}
\end{figure}
\begin{lstlisting}[float, caption=Example explaining rule assn-2,label=lst2]
if ($x'$)
  $z = y_1$@\label{hm1}@
else
  $z = y_2$@\label{hm2}@
if ($x_1$)@\label{bl1}@
  $z = x_1$@\label{hl1}@
if ($\texttt{not}(x_2)$)@\label{if2}@
  $z = x_2$@\label{hl2}@
if ($z$)@\label{if3}@
  $w = z$@\label{hl3}@
\end{lstlisting}
\begin{figure}
\centering
\begin{tikzpicture}
\node (H) at (0,3) {$H$};
\node (M1) at (-1,2) {$M_1$};
\node (M2) at (1,2) {$M_2$};
\node (Lp) at (0,1) {$L'$};
\node (L1) at (-1,1) {$L_1$};
\node (L2) at (1,1) {$L_2$};
\node (L) at (0,0) {$L$};
	
\tikzstyle{every path}=[black] 
\path		(H) edge (M1)
			(H) edge (M2)
			(M1) edge (L1)
			(M1) edge (Lp)
			(M2) edge (Lp)
			(M2) edge (L2)
			(L1) edge (L)
			(Lp) edge (L)
			(L2) edge (L);				
\end{tikzpicture}
\caption{Lattice explaining rule assn-2}\label{fig:lattice}
\end{figure}
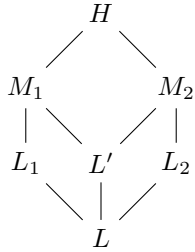

\begin{table*}
\centering
\begin{tabular} {|l||@{\,}c@{\,}||@{\,}c@{\,}|@{\,}c@{\,}|}
\hline
&
\multicolumn{3}{c|}{$w = \texttt{false}^{L_1},\ x_1 =
  \texttt{true}^{L_1},\ y_1 = \texttt{false}^{M_1},\ y_2 =
  \texttt{true}^{M_2}$}
\\
\cline{2-4}
&
$x' = \texttt{true}^{L'}$
&
\multicolumn{2}{c|}{$x' = \texttt{false}^{L'}$}
\\
&
$x_2 = \texttt{true}^{L_2}$
&
\multicolumn{2}{c|}{$x_2 = \texttt{false}^{L_2}$}
\\
\cline{3-4}
&
&
assn-2 with $k:= \lab_x\pl$
&
assn-2 with $k:= ((\pc \sqcup m) \sqcap \lab_x)\pl$
\\
\hline
\texttt{if} ($x'$)&if-branch taken, $\pc = L'$&&\\
\quad$z = y_1$&$z = \texttt{false}^{M_1}$&&\\
\texttt{else}&&else-branch taken, $\pc = L'$&else-branch taken, $\pc = L'$\\
\quad$z = y_2$&&$z = \texttt{true}^{M_2}$&$z = \texttt{true}^{M_2}$\\
\texttt{if} ($x_1$)&branch taken, $\pc = L_1$&branch taken, $\pc = L_1$&branch taken, $\pc = L_1$\\
\quad$z = x_1$&$z = \texttt{true}^{L_1}$&$z = \texttt{true}^{M_2\pl}$&$z = \texttt{true}^{L\pl}$\\
\texttt{if} ($\texttt{not}(x_2)$)&branch not taken&branch taken, $\pc = L_2$&branch taken, $\pc = L_2$\\
\quad$z = x_2$& 
&$z = \texttt{false}^{L_2}$&$z = \texttt{false}^{L\pl}$\\
\texttt{if} ($z$)&branch taken, $\pc = L_1$&branch not taken&execution halted\\
\quad$w = z$&$w = \texttt{true}^{L_1}$&&\\
\hline
Result & $w = \texttt{true}^{L_1}$ & $w =
\texttt{false}^{L_1}$ (leak) & execution halted (no leak)\\
\hline
\end{tabular}
\caption{Execution steps in two runs of the program from Listing~\ref{lst2}, with two variants of the rule assn-2}
\label{tblassn}
\end{table*}

The syntax of labels is listed in Figure~\ref{fig:labels}. Labels
$k,l,m$ may be lattice elements $\lab$ or $\star$-ed lattice elements
$\lab\pl$. In examples, we use suggestive lattice element names $L, M,
H$ (low, medium, high). Labels of the form $\lab$ are called
$\star$-free or \emph{pure}. Figure~\ref{fig:labels} also defines the
join operation $\sqcup$ on labels, which is used to combine labels of
the arguments of $\odot$. This definition is based on the intuition
above. When the two operands of $\odot$ are labeled $\lab_1$ and
$\lab_2\pl$, $\lab_1 \sqcup \lab_2$ is a lower bound on the pure label
of the resulting value in any execution (because $\lab_2$ is a lower
bound on the pure label of $\lab_2\pl$ in any run). Hence, $\lab_1
\sqcup \lab_2\pl = (\lab_1 \sqcup \lab_2)\pl$. The reason for the
definition $\lab_1\pl \sqcup \lab_2\pl = (\lab_1 \sqcup \lab_2)\pl$ is
similar.

Our rules for assignment are shown in Figure~\ref{fig:assn-our}. They
strictly generalize the rule assn-PUS for the two-point lattice,
treating $P = L\pl$. Rule assn-1 applies when the existing label of
the variable being assigned to is $\lab_x$ or $\lab_x\pl$ and $\pc
\sqsubseteq \lab_x$. The key intuition behind the rule is the
following: If $\pc \sqsubseteq \lab_x$, then it is safe to overwrite
the variable, because $\lab_x$ is necessarily a lower bound on the
(pure) label of $x$ in this and any alternate execution (see the
\framebox{framebox} above). Hence, overwriting the variable cannot
cause an implicit flow. As expected, the label of the overwritten
variable is $\pc \sqcup m$, where $m$ is the label of the value
assigned to $x$.

\footnotetext{In the original paper, $k := (\pc \sqcap \lab_x
  )\pl$ in the rule for assn-2}

Rule assn-2 applies in the remaining case --- when $\pc
\not\sqsubseteq \lab_x$. In this case, there may be an implicit flow,
so the final label on $x$ must have the form $\lab\pl$ for some
$\lab$. The question is which $\lab$? Intuitively, it may seem that
one could choose $\lab = \lab_x$, the pure part of the original label
of $x$. The final label on $x$ would be $\lab_x\pl$ and this would
satisfy the intuitive meaning of $\star$ written in the
\framebox{framebox} above. Indeed, this intuition suffices for the
two-point lattice of Section~\ref{sec:existing}. However, for a more
general lattice, this intuition is unsound, as we illustrate with an
example below. The correct label is $((\pc \sqcup m) \sqcap \lab_x)\pl$. (Note
that this correct label is independent of the label $m$ of the value
assigned to $x$. This is sound because $x$ is $\star$-ed and cannot be
case-analyzed later, so the label on the value in it is irrelevant.)

\paragraph{Example} 
We illustrate why we need the label $k := ((\pc \sqcup m) \sqcap \lab_x)\pl$
instead of $k := \lab_x\pl$ in rule assn-2. Consider the lattice of
Figure~\ref{fig:lattice} and the program of Listing~\ref{lst2}. Assume
that, initially, the variables $z$, $w$, $x_1$, $x'$, $x_2$, $y_1$ and
$y_2$ have labels $H$, $L_1$, $L_1$, $L'$, $L_2$, $M_1$ and $M_2$,
respectively. Fix the attacker at level $L_1$. Fix the value of $x_1$
at $\texttt{true}^{L_1}$, so that the branch on line~\ref{bl1} is
always taken and line~\ref{hl1} is always executed. Set $y_1 \mapsto
\texttt{false}^{M_1}, y_2 \mapsto \texttt{true}^{M_2}, w \mapsto
false^{L_1}$ initially. The initial value of $z$ is
irrelevant. Consider two executions of the program starting from two
stores $\sigma_1$ with $x' \mapsto \texttt{true}^{L'}, x_2 \mapsto
\texttt{true}^{L_2}$ and $\sigma_2$ with $x' \mapsto
\texttt{false}^{L'}, x_2 \mapsto \texttt{false}^{L_2}$. Note that
because $L'$ and $L_2$ are incomparable to $L_1$ in the lattice,
$\sigma_1$ and $\sigma_2$ are equivalent for $L_1$. 

We show that requiring $k := \lab_x\pl$ in rule assn-2 causes an
implicit flow that is observable for $L_1$. The intermediate values
and labels of the variables for executions starting from $\sigma_1$
and $\sigma_2$ are shown in the second and third columns of
Table~\ref{tblassn}. Starting with $\sigma_1$, line~\ref{hm1} is
executed, but line~\ref{hm2} is not, so $z$ ends with
$\texttt{false}^{M_1}$ at line~\ref{bl1} (rule assn-1 applies at
line~\ref{hm1}). At line~\ref{hl1}, $z$ contains $\texttt{true}^{L_1}$
(again by rule assn-1) and line~\ref{hl2} is not executed. Thus, the
branch on line~\ref{if3} is taken and $w$ ends with
$\texttt{true}^{L_1}$ at line~\ref{hl3}. Starting with $\sigma_2$,
line~\ref{hm1} is not executed, but line~\ref{hm2} is, so $z$ becomes
$\texttt{true}^{M_2}$ at line~\ref{bl1} (rule assn-1 applies at
line~\ref{hm2}). At line~\ref{hl1}, rule assn-2 applies, but because
we assume that $k := \lab_x\pl$ in that rule, $z$ now contains the
value $\texttt{true}^{M_2\pl}$. As the branch on line~\ref{if2} is
taken, at line~\ref{hl2}, $z$ becomes $\texttt{false}^{L_2}$ by rule
assn-1 because $L_2 \sqsubseteq M_2$. Thus, the branch on
line~\ref{if3} is not taken and $w$ ends with $\texttt{false}^{L_1}$
in this execution. Hence, $w$ ends with $\texttt{true}^{L_1}$ and
$\texttt{false}^{L_1}$ in the two executions, respectively. The
attacker at level $L_1$ can distinguish these two results; hence, the
program leaks the value of $x'$ and $x_2$ to $L_1$.

With the correct assn-2 rule in place, this leak is avoided (last
column of Table~\ref{tblassn}). In that case, after the assignment on
line~\ref{hl1} in the second execution, $z$ has label $((L_1 \sqcup L_1) \sqcap
M_2)\pl = L\pl$. Subsequently, after line~\ref{hl2}, $z$ gets the
label $L\pl$. As case analysis on a $\star$-ed value is not allowed,
the execution is halted on line~\ref{if3}. This guarantees
termination-insensitive noninterference with respect to the attacker
at level $L_1$.

\subsection{Store equivalence}

To prove noninterference for our generalized permissive-upgrade, we
define equivalence of labeled values relative to an adversary at
arbitrary lattice level $\lab$. The definition is shown below. We
explain later how it is obtained, but we note that clauses (3)--(5)
here refine clause (3) of Definition~\ref{def:eq-existing} for the
two-point lattice. The obvious generalization of clause (3) of
Definition~\ref{def:eq-existing} --- $n_1^k \sim_\lab n_2^m$ whenever
either $k$ or $m$ is $\star$-ed --- is too coarse to allow us to prove
noninterference inductively. For the degenerate case of the two-point
lattice, this definition also degenerates to
Definition~\ref{def:eq-existing} (there, $\lab$ is fixed at $L$, $P =
L\pl$ and only $L$ may be $\star$-ed).

\begin{mydef}
\label{def1}
Two values $n_1^k$ and $n_2^m$ are $\lab$-equivalent, written $n_1^k
\sim_\lab n_2^m$, iff either
\begin{enumerate}
\item $k = m = \lab' \sqsubseteq \lab$ and $n_1 = n_2$, or
\item $ k = \lab'
  \not\sqsubseteq \lab$ and $m = \lab'' \not\sqsubseteq \lab$, or 
\item $k = \lab_1\pl$ and $m = \lab_2\pl$, or
\item $k = \lab_1\pl$ and $m = \lab_2$ and ($\lab_2 \not\sqsubseteq
  \lab$ or $\lab_1 \sqsubseteq \lab_2 $), or
\item $k = \lab_1$ and $m = \lab_2\pl$ and ($\lab_1 \not\sqsubseteq
  \lab$ or $\lab_2 \sqsubseteq \lab_1$)
\end{enumerate}
\end{mydef}

We obtained this definition by constructing (through examples) an
extensive transition graph of pairs of labels that may be assigned to
a single variable at corresponding program points in two executions of
the same program. Our starting point is label-pairs of the form
$(\lab, \lab)$. We discovered that this characterization of
equivalence is both sufficient and necessary. It is sufficient in the
sense that it allows us to prove TINI inductively. It is necessary in
the sense that example programs can be constructed that end in states
exercising every possible clause of this definition. A technical
appendix, available from the authors' homepages, lists these examples.

\subsection{Termination-Insensitive Noninterference}

Using the above definition of equivalence of labeled values, we can
prove TINI for our generalized permissive-upgrade. A significant
difficulty in proving the theorem is that our definition of
$\sim_\lab$ is not transitive. The same problem arises for the
two-point lattice in~\cite{plas10}. There, the authors resolve the
issue by defining a special relation called evolution. Here, we follow
a more conventional approach based on the standard confinement
lemma. The need for evolution is averted using several auxiliary
lemmas that we list below. Detailed proofs of all lemmas and theorems
are presented in our technical appendix.

\begin{myLemma}[Expression evaluation]
\label{exp}
If $\langle e, \sigma_1 \rangle \Downarrow n_1^{k_1}$ and $\langle e,
\sigma_2 \rangle \Downarrow n_2^{k_2}$ and $\sigma_1 \sim_\lab \sigma_2$,
then $n_1^{k_1} \sim_\lab n_2^{k_2}$.
\end{myLemma}
\begin{proof}
By induction on $e$.
\end{proof}


\begin{myLemma}[$\star$-preservation]
\label{sup1}
If $\langle c, \sigma \rangle \Downarrow_\pc \sigma'$ and
$\Gamma(\sigma(x)) = \lab\pl $ and $\pc \not\sqsubseteq \lab$, then
$\Gamma(\sigma'(x)) = \lab'\pl$ and $\lab' \sqsubseteq \lab$.
\end{myLemma}
\begin{proof}
By induction on the given derivation.
\end{proof}

\begin{mycor}
\label{cor1}
If $\langle c, \sigma \rangle \Downarrow_\pc \sigma'$ and
$\Gamma(\sigma(x)) = \lab\pl $ and $\Gamma(\sigma'(x)) = \lab'$, then
$\pc \sqsubseteq \lab$.
\end{mycor}
\begin{proof}
Immediate from Lemma~\ref{sup1}.
\end{proof}

\begin{myLemma}[$\pc$-lemma]
\label{pcl}
If $\langle c, \sigma \rangle \Downarrow_\pc \sigma'$ and
$\Gamma(\sigma'(x)) = \lab$, then $\sigma(x) = \sigma'(x)$ or
$\pc \sqsubseteq \lab$.
\end{myLemma}
\begin{proof}
By induction on the given derivation.
\end{proof}

\begin{mycor}
\label{cor2}
If $\langle c, \sigma \rangle \Downarrow_\pc \sigma'$ and
$\Gamma(\sigma(x)) = \lab\pl $ and $\Gamma(\sigma'(x)) = \lab'$, then
$\pc \sqsubseteq \lab'$.
\end{mycor}
\begin{proof}
Immediate from Lemma~\ref{pcl}.
\end{proof}

Using these lemmas, we can prove the standard confinement lemma and
noninterference. 

\begin{myLemma}[Confinement Lemma]
\label{conf}
If $\pc \not\sqsubseteq \lab$ and $\langle c, \sigma \rangle
\Downarrow_\pc \sigma'$, then $\sigma \sim_\lab \sigma'$.
\end{myLemma}
\begin{proof}
By induction on the given derivation.
\end{proof}

\begin{myThm}[TINI for generalized permissive-upgrade]
  If $~\sigma_1 \sim_\lab \sigma_2$ and $\langle c, \sigma_1 \rangle
  \Downarrow_\pc \sigma_1' $ and $\langle c, \sigma_2
  \rangle\Downarrow_\pc \sigma_2' $, then $\sigma_1' \sim_\lab
  \sigma_2'$.
\end{myThm}
\begin{proof} By induction on $c$.
\end{proof}



\subsection{Incomparability to the Generalization of Section~\ref{sec:existing}}


We have two distinct and sound generalizations of the original
permissive-upgrade for the two-point lattice: The generalization to
pointwise products of two-point lattices or, equivalently, to powerset
lattices as described in Section~\ref{sec:existing}, and the
generalization to arbitrary lattices described earlier in this
section. For brevity, we call these generalizations puP
(Section~\ref{sec:existing}) and puA (Section~\ref{sec:gen:pus}),
respectively (P and A stand for \underline{p}owerset and
\underline{a}rbitrary, respectively). Since both puP and puA apply to
powerset lattices, an obvious question is whether one is more
permissive than the other on such lattices. 
We show here that the
permissiveness of puP and puA on powerset lattices is
\emph{incomparable} --- there are examples on which each
generalization is more permissive than the other. We explain one
example in each direction below. Roughly, incomparability exists
because puP tracks finer taints (it tracks partial leaks for each
principal separately), but puA's rules for overwriting
partially-leaked variables are more permissive.

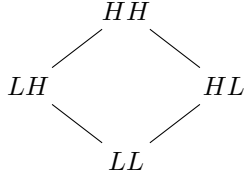
\begin{figure}
\centering
\begin{tikzpicture}[xscale=1.3]
\node (Lp) at (0,2) {$HH$};
\node (L1) at (-1,1) {$LH$};
\node (L2) at (1,1) {$HL$};
\node (L) at (0,0) {$LL$};
	
\tikzstyle{every path}=[black] 
\path 	(L2) edge (L)
			(L1) edge (L)
			(Lp) edge (L1)
			(Lp) edge (L2);				
\end{tikzpicture}
\caption{A powerset/product lattice}\label{fig:lattice1}
\end{figure}

	

We use the powerset lattice of Figure~\ref{fig:lattice1} for our
example. This lattice is the pointwise lifting of the order $L
\sqsubset H$ to the set $S = \{L, H\} \times \{L, H\}$. For brevity,
we write this lattice's elements as $LL$, $LH$, etc. When puP is
applied to this lattice, labels are drawn from the set $\{L, H, P\}
\times \{L, H, P\}$. We write these labels concisely as $LP$, $HL$,
etc. For puA, labels are drawn from the set $S \cup S\pl$. We write
these labels $LH$, $LH\pl$, etc. Note that $LH\pl$ parses as
$(LH)\pl$, not $L(H\pl)$ (the latter is not a valid label in puA
applied to this lattice).

\begin{lstlisting}[float,caption=Example where puA is more permissive than puP,label=list1]
if ($y$)@\label{bi1}@
  $z = 2$ @\label{li1}@
$x = y+z$ @\label{li2}@
if ($y$)@\label{bi2}@
  $x = 3$ @\label{li3}@
if ($x$) @\label{li4}@
  $y = 5$
\end{lstlisting}

\paragraph{Example}



We start with an example program which executes completely under puA,
but gets stuck under puP (since puA is sound, there is no actual
information leak in the program). This example is shown in
Listing~\ref{list1}. Assume that $x$, $y$ and $z$ have initial labels
$LL$, $HH$ and $LH$, respectively and that $y \mapsto
\texttt{true}^{HH}$, so the branches on lines~\ref{bi1} and~\ref{bi2}
are both taken. The initial values of $x$ and $z$ are irrelevant but
their labels are relevant.

Under puP, $z$ obtains label $PH$ at line~\ref{li1} by rule
assn-PUS'. At line~\ref{li2}, $x$ obtains the label $ (HH) \sqcup (PH)
= PH$. At line~\ref{li3}, the label of $x$ stays $PH$ by rule
assn-PUS'. At line~\ref{li4}, the program halts because the branch
condition $x$'s label contains $P$.

On the other hand, under puA, the program executes to completion. At
line~\ref{li1}, $z$ obtains the label $(((HH) \sqcup (LL)) \sqcap (LH))\pl = LH\pl$
by rule assn-2. At line~\ref{li2}, $x$ obtains the label $(HH) \sqcup
(LH\pl) = HH\pl$. At line~\ref{li3}, the label of $x$ changes to
$HH$: the $pc$ at this point (equal to the label of $y$) is $HH$, so
rule assn-1 applies. Since $HH$ is pure, the program does not stop at
line~\ref{li4}.

Hence, on this example, puA is more permissive than puP. 

\begin{lstlisting}[float,caption=Example where puP is more permissive than puA,label=list3]
if ($y$) @\label{li32}@
  $x = z$ @\label{li33}@
if ($z$) @\label{li34}@
  $x = z$ @\label{li35}@
if ($x$)  @\label{li36}@
  $z = x$
\end{lstlisting}

\paragraph{Example}
Next, consider the program in Listing~\ref{list3}. For this program,
puP is more permissive than puA. Assume that $x$, $y$ and $z$ have
initial labels $LL$, $HL$ and $LH$, respectively and that the initial
store contains $y \mapsto \texttt{true}^{HL}, z \mapsto
\texttt{true}^{LH}$, so the branches on lines~\ref{li32}
and~\ref{li34} are both taken. The initial value in $x$ is irrelevant
but its label is important.

Under puA, $x$ obtains label $(((HL) \sqcup (LH)) \sqcap (LL))\pl = LL\pl$ at
line~\ref{li33} by rule assn-2. At line~\ref{li35}, the same rule
applies but the label of $x$ remains $LL\pl$. When the program tries
to branch on $x$ at line~\ref{li36}, it is stopped.

In contrast, under puP, this program executes to completion. At
line~\ref{li33}, the label of $x$ changes to $PH$ by rule
assn-PUS'. At line~\ref{li35}, the label changes to $LH$ because $pc$
and the label of $z$ are both $LH$. Since this new label has no $P$,
line~\ref{li36} executes without halting.

Hence, for this example, puP is more permissive than puA.

%% file: related.tex
\section{Related Work}

We directly build on, and generalize, the permissive-upgrade check of
Austin and Flanagan~\cite{plas10}. Earlier sections describe the
connection of that work to ours.  In recent work, we implemented the
permissive-upgrade check for JavaScript's bytecode in the WebKit
browser engine~\cite{POST14}. Our formalization in that work is
limited to the two-point lattice, and generalizing that formalization
motivated this paper. In working with JavaScript bytecode, we found
permissive-upgrade indispensable: The source-to-bytecode compiler in
WebKit generates assignments to dead variables under high $pc$, which
halts program execution if the no-sensitive-upgrade check (NSU) is used
instead of permissive-upgrade.

The permissive-upgrade check is just one of many ways of avoiding
implicit flows in dynamic IFC when labels on variables are
flow-sensitive (not fixed upfront).  A pre-cursor to the
permissive-upgrade is the NSU check, first proposed by
Zdancewic~\cite{zdancewic02PhD}. A different way of handling the
problem of implicit flows through flow-sensitive labels is to assign a
(fixed) label to each label; this approach has been examined in recent
work by Buiras \emph{et al.} in the context of a language with a
dedicated monad for tracking information flows~\cite{CSF14}. The
precise connection between that approach and permissive-upgrade
remains unclear, although Buiras \emph{et al.}\ sketch a technique
related to permissive-upgrade in their system, while also noting that
generalizing permissive-upgrade to arbitrary lattices is
non-obvious. Our work confirms the latter and shows how it can be
done.

Birgisson \emph{et al.}~\cite{esorics12} describe a testing-based
approach that adds variable upgrade annotations to avoid halting on
the NSU check in an implementation of dynamic IFC for
JavaScript~\cite{csf12}. Hritcu \emph{et al.}  improve permissiveness
by making IFC errors recoverable in the language
Breeze~\cite{Hritcu:ifc}. This is accomplished by a combination of two
methods: making all labels public (by upgrading them when necessary in
a public $\pc$) and by delaying exceptions. 


Finally, IFC with flow-sensitive labels can be enforced statically or
using hybrid techniques that combine static and dynamic
methods~\cite{hunt2006:types,russo10:dyn}. Russo \emph{et
  al.}~\cite{russo10:dyn} show formally that the expressive power of
sound flow-sensitive static analysis and sound flow-sensitive dynamic
monitors is incomparable. Hence, there is merit to investigating
hybrid approaches.




%% file: conclusion.tex
\section{Conclusion}

Permissive-upgrade is a useful technique for avoiding implicit flows
in dynamic information flow control. However, the technique's initial
development was limited to a two-point lattice and pointwise products
of such lattices. We show, by construction, that permissive-upgrade
can be generalized to arbitrary lattices and that the generalization's
rules and correctness definitions are non-trivial.

\section*{Acknowledgments}
 We thank anonymous reviewers for their excellent feedback on
this paper's presentation. This work was funded in part by the
Deutsche For\-schungs\-ge\-mein\-schaft (DFG) grant “Information Flow
Control for Browser Clients” under the priority program “Reliably
Secure Software Systems” (RS$^3$), and the German Federal Ministry of
Education and Research (BMBF) within the Center for IT-Security,
Privacy and Accountability (CISPA) at Saarland University.

%% file: appendix.tex
\setcounter{myThm}{0}
\setcounter{myLemma}{0}

\lstset{numbers=none}
\begin{table*}
\centering
\begin{tabular} {|c|c|c|c|c|c|c|c|}
\hline
 & $\ell, \ell$ & $\ell_1\pl, \ell_2 $ & $\ell_1, \ell_2\pl$ &
 $\ell_1\pl, \ell_2\pl $& $\textcolor{red}{\ell_1},
 \textcolor{red}{\ell_2}$ & $\ell_1 \pl,
 \textcolor{red}{\ell_2}$& 
 $\textcolor{red}{\ell_1}, \ell_2 \pl$ \\
\hline
$\ell, \ell$ & 
-
&
\begin{lstlisting}
if(h) 
 x1 = l
\end{lstlisting} & 
\begin{lstlisting}
if(h) 
 x1 = l
\end{lstlisting} & 
\begin{lstlisting}
if(h) 
 x1 = l
else
 x1 = l
\end{lstlisting} 
& 
\begin{lstlisting}
x1 = h
\end{lstlisting} 
& 
\begin{lstlisting}
x1 = m
if(h) 
 x1 = 4
if(m)
  x1 = l*
\end{lstlisting} 
& 
\begin{lstlisting}
x1 = m
if(h) 
 x1 = 4
if(m)
  x1 = l*
\end{lstlisting} 
\\
\hline
$\ell_1\pl, \ell_2$ & 
\begin{lstlisting}
x1 = l
\end{lstlisting} & 
- & 
\begin{lstlisting}
x1 = l
if(h)
  x1 = l
\end{lstlisting} & 
\begin{lstlisting}
if(h)
  x1 = l
\end{lstlisting} & 
\begin{lstlisting}
x1 = h
\end{lstlisting} & 
\begin{lstlisting}
x1 = m 
if (h)
  x1 = l
if(m)
  x1 = l*
\end{lstlisting} & 
\begin{lstlisting}
x1 = m 
if (h)
  x1 = l
if(m)
  x1 = l*
\end{lstlisting} \\
\hline
$\ell_1, \ell_2\pl$ &
\begin{lstlisting}
x1 = l
\end{lstlisting} & 
\begin{lstlisting}
x1 = l
if(h)
  x1 = l
\end{lstlisting} & 
- & 
\begin{lstlisting}
if(h)
  x1 = l
\end{lstlisting} & 
\begin{lstlisting}
x1 = h
\end{lstlisting} & 
\begin{lstlisting}
x1 = m 
if (h)
  x1 = l
if(m)
  x1 = l*
\end{lstlisting} & 
\begin{lstlisting}
x1 = m 
if (h)
  x1 = l
if(m)
  x1 = l*
\end{lstlisting} \\
\hline
$\ell_1\pl, \ell_2\pl$ & 
\begin{lstlisting}
x1 = l
\end{lstlisting} 
& 
\begin{lstlisting}
x1 = l
if (h)
 x1 = l
\end{lstlisting} 
& 
\begin{lstlisting}
x1 = l
if (h)
 x1 = l
\end{lstlisting} 
& 
-
&
\begin{lstlisting}
 x1 = h
\end{lstlisting} 
& 
\begin{lstlisting}
x1 = m
if (h)
 x1 = l
if (m)
 x1 = l*
\end{lstlisting} 
&
\begin{lstlisting}
x1 = m
if (h)
 x1 = l
if (m)
 x1 = l*
\end{lstlisting}  
\\
\hline
$\textcolor{red}{\ell_1}, \textcolor{red}{\ell_2}$ & 
\begin{lstlisting}
x1 = l 
\end{lstlisting}  
& 
\begin{lstlisting}
x1 = l 
if (h)
  x1 = l
\end{lstlisting}  
& 
\begin{lstlisting}
x1 = l 
if (h)
  x1 = l
\end{lstlisting}  
& 
\begin{lstlisting}
x1 = l 
if (h)
  x1 = l
else
  x1 = l
\end{lstlisting}  
& 
-
&
\begin{lstlisting}
x1 = m
if (h)
  x1 = l
if(m)
  x1 = l*
\end{lstlisting}   
& 
\begin{lstlisting}
x1 = m
if (h)
 x1 = l
if (m)
 x1 = l*
\end{lstlisting}   
\\
\hline
$\ell_1 \pl, \textcolor{red}{\ell_2}$ & 
\begin{lstlisting}
x1 = l
\end{lstlisting}
&
\begin{lstlisting}
x1 = l
if (h)
  x1 = l 
\end{lstlisting} 
& 
\begin{lstlisting}
x1 = l
if (h)
  x1 = l 
\end{lstlisting} 
& 
\begin{lstlisting}
x1 = l
if (h)
  x1 = l 
else
  x1 = l
\end{lstlisting} 
& 
\begin{lstlisting}
x1 = h
\end{lstlisting} 
&
- 
&
\begin{lstlisting}
x1 = m
if (h)
 x1 = l
if (m)
 x1 = l*
\end{lstlisting} 
\\
\hline
$\textcolor{red}{\ell_1}, \ell_2 \pl$ &
\begin{lstlisting}
x1 = l
\end{lstlisting} 
& 
\begin{lstlisting}
x1 = l
if (h)
  x1 = l 
\end{lstlisting} 
& 
\begin{lstlisting}
x1 = l
if (h)
  x1 = l 
\end{lstlisting} 
& 
\begin{lstlisting}
x1 = l
if (h)
  x1 = l 
else
  x1 = l
\end{lstlisting} 
& 
\begin{lstlisting}
x1 = h
\end{lstlisting} 
& 
\begin{lstlisting}
x1 = m
if (h)
  x1 = l
if(m)
  x1 = l* 
\end{lstlisting} 
& 
-
\\
\hline
\end{tabular}
\caption{Examples for all possible transitions of low-equivalent to
  low-equivalent values}\label{tab:flows}
\end{table*}

\appendix
\section{Appendix}
Assumptions: \\
l is a variable with label L \\
m is a variable with label M \\
h is a variable with label H \\
l\pl is a variable with label L\pl \\
L $\sqsubseteq$ M $\sqsubseteq$ H\\
and we assume the attacker at level L. \textcolor{red}{$\ell$}
represents the labels that are above the level of the attacker.

The table shows example programs for the transition from
low-equivalent values to low-equivalent values. First column and first
row of the table represents all the possible ways in which two values
can be low-equivalent (from defintion~\ref{def1}).

\subsection{Proofs and Results}
\begin{myLemma}{\emph{Expression Evaluation Lemma}}\\
\label{exp}
If $\sigma_1 \sim_\lab \sigma_2$, \\
$\langle e, \sigma_1 \rangle \Downarrow n_1^{k_1}$, \\
$\langle e, \sigma_2 \rangle \Downarrow n_2^{k_2}$, \\
then
$n_1^{k_1} \sim_\lab n_2^{k_2}$.
\end{myLemma}
\begin{proof}
Proof by induction on the derivation and case analysis on the last
expression rule.
\begin{enumerate}
\item const: $n_1 = n_2 = n$ and $k_1 = k_2 = \perp$. 

\item var: As $\sigma_1 \sim_\lab \sigma_2$, $\forall x.\sigma_1(x) = n_1^{k_1}
  \sim_\lab \sigma_2(x) = n_2^{k_2}$. 

\item oper: IH1: If $\langle e_1, \sigma_1 \rangle \Downarrow
  n_1'^{k_1'}$, $\langle e_1, \sigma_2 \rangle \Downarrow
  n_2'^{k_2'}$, $\sigma_1 \sim_\lab \sigma_2$, then $n_1'^{k_1'}
  \sim_\lab n_2'^{k_2'}$.\\
IH2: If $\langle e_2, \sigma_1 \rangle \Downarrow
  n_1''^{k_1''}$, $\langle e_2, \sigma_2 \rangle \Downarrow
  n_2''^{k_2''}$, $\sigma_1 \sim_\lab \sigma_2$, then $n_1''^{k_1''}
  \sim_\lab n_2''^{k_2''}$.\\
T.S. $n_1^{k_1} \sim_\lab n_2^{k_2}$, where $n_1 = n_1' \odot n_1''$, $n_2 = n_2' \odot n_2''$ 
and $k_1 = k_1' \sqcup k_1''$, $k_2 = k_2' \sqcup k_2''$.\\
As $\sigma_1 \sim_\lab \sigma_2$, from IH1 and IH2, $n_1'^{k_1'}
  \sim_\lab n_2'^{k_2'}$ and $n_1''^{k_1''}  \sim_\lab n_2''^{k_2''}$.\\
Proof by case analysis on low-equivalence definition for $n_1'^{k_1'} 
  \sim_\lab n_2'^{k_2'}$ followed by case analysis on low-equivalence definition for $n_1''^{k_1''}
  \sim_\lab n_2''^{k_2''}$.
\end{enumerate}
\end{proof}

\begin{myLemma}{\emph{$\star$-preservation Lemma}}\\
\label{sup1}
$\forall x$.If $\langle c, \sigma \rangle \Downarrow_\pc
\sigma'$, $\Gamma(\sigma(x)) =   \lab \pl
 $ and $\pc \not\sqsubseteq \lab$, then $\Gamma(\sigma'(x)) =
  \lab' \pl   \wedge \lab' \sqsubseteq \lab$
\end{myLemma}
\begin{proof}
Proof by induction on the derivation and case analysis on the last
rule.

\begin{enumerate}
\item skip : $\sigma = \sigma'$.

\item assn-1: As $\pc \not\sqsubseteq \lab$, these cases do not apply.

\item assn-2: From the premises, for $x$ in statement $c$,
  $\Gamma(\sigma'(x)) =   ((\pc \sqcup m)
  \sqcap \lab) \pl  =  \lab'$. Thus, $\lab' \sqsubseteq \lab$.\\
  For any other $y$, $\sigma(y) = \sigma'(y)$. Thus, $\lab' = \lab$.

\item seq : IH1 : $\forall x$.If $ \langle c, \sigma \rangle \Downarrow_\pc
\sigma''$, $\Gamma(\sigma(x)) =   \lab\pl
 $ and $\pc \not\sqsubseteq \lab$, then $\Gamma(\sigma''(x)) =
  \lab'' \pl   \wedge \lab'' \sqsubseteq \lab$\\
IH2 : $\forall x$.If $ \langle c, \sigma'' \rangle \Downarrow_\pc
\sigma'$, $\Gamma(\sigma''(x)) =   \lab'' \pl
 $ and $\pc \not\sqsubseteq \lab''$, then $\Gamma(\sigma'(x)) =
  \lab' \pl   \wedge \lab' \sqsubseteq \lab''$\\
Thus, from IH1 and IH2, $\Gamma(\sigma'(x)) =
  \lab' \pl   \wedge \lab' \sqsubseteq \lab$.

\item if-else: Let $k = \lab''$. \\
 IH: $\forall x$.If $ \langle c, \sigma \rangle \Downarrow_{\pc \sqcup \lab''}
\sigma'$, $\Gamma(\sigma(x)) =   \lab \pl
 $ and $\pc \sqcup \lab'' \not\sqsubseteq \lab$, then $\Gamma(\sigma'(x)) =
  \lab' \pl   \wedge \lab' \sqsubseteq \lab$\\
 As $\pc \not\sqsubseteq \lab$, so $\pc \sqcup \lab'' \not\sqsubseteq
 \lab$. \\ Thus from IH, $\Gamma(\sigma'(x)) =
  \lab' \pl   \wedge \lab' \sqsubseteq \lab$

\item while-t: Let $k = \lab_e$. \\
 IH1: $\forall x$.If $\langle c, \sigma \rangle \Downarrow_{\pc \sqcup \lab_e}
\sigma''$, $\Gamma(\sigma(x)) =   \lab \pl
 $ and $\pc \sqcup \lab_e \not\sqsubseteq \lab$, then $\Gamma(\sigma''(x)) =
  \lab''\pl   \wedge \lab'' \sqsubseteq \lab$\\
IH2: $\forall x$.If $\langle c, \sigma'' \rangle \Downarrow_{\pc \sqcup \lab_e}
\sigma'$, $\Gamma(\sigma''(x)) =   \lab'' \pl
 $ and $\pc \sqcup \lab_e \not\sqsubseteq \lab$, then $\Gamma(\sigma'(x)) =
  \lab' \pl   \wedge \lab' \sqsubseteq \lab$\\
 As $\pc \not\sqsubseteq \lab$, so $\pc \sqcup \lab_e \not\sqsubseteq
 \lab$. \\ Thus from IH1 and IH2, $\Gamma(\sigma'(x)) =
  \lab' \pl   \wedge \lab' \sqsubseteq \lab$

\item while-f : $\sigma = \sigma'$.
\end{enumerate}
\end{proof}

\begin{myLemma}{\emph{$\pc$ Lemma}}\\
If $\langle c, \sigma \rangle \Downarrow_\pc \sigma'$, then $\forall
x.\Gamma(\sigma'(x)) = \lab \implies (\sigma(x) = \sigma'(x)) \vee \pc
\sqsubseteq \lab$.
\end{myLemma}
\begin{proof}
Proof by induction on the derivation and case analyis on the last
rule.
\begin{itemize}
  \item skip: $\sigma(x) = \sigma'(x)$. 
  \item assn: For $x$ in the statement $c$, by premises, $\lab = \pc
    \sqcup \lab_e$. Thus, $\pc \sqsubseteq \lab$.\\
    For any other $y$ s.t. $\Gamma(\sigma'(y)) = \lab'$, $\sigma(y) =
    \sigma'(y)$.
  \item seq: IH1: If $\langle c, \sigma \rangle \Downarrow_\pc \sigma''$, then $\forall
x.\Gamma(\sigma''(x)) = \lab'' \implies (\sigma(x) = \sigma''(x)) \vee \pc
\sqsubseteq \lab''$.\\
IH2: If $\langle c, \sigma'' \rangle \Downarrow_\pc \sigma'$, then $\forall
x.\Gamma(\sigma'(x)) = \lab \implies (\sigma''(x) = \sigma'(x)) \vee \pc
\sqsubseteq \lab$.\\
  From IH2, if $\sigma''(x) \neq \sigma'(x)$, then $\pc \sqsubseteq
  \lab$.\\
  If $\sigma''(x) = \sigma'(x)$, then from IH1:
  \begin{itemize}
   \item If $\sigma(x) = \sigma''(x)$: $\sigma(x) = \sigma'(x)$.
   \item If $\sigma(x) \neq \sigma''(x)$: $\pc \sqsubseteq \lab''$,
     where $\lab'' = \Gamma(\sigma''(x))$. As $\sigma''(x) =
     \sigma'(x)$, $\lab'' = \Gamma(\sigma'(x)) = \lab$. Thus, $\pc
     \sqsubseteq \lab$.
   \end{itemize}

 \item if-else: IH: If $\langle c, \sigma \rangle \Downarrow_{\pc \sqcup \lab_e} \sigma'$, then $\forall
x.\Gamma(\sigma'(x)) = \lab \implies (\sigma(x) = \sigma'(x)) \vee \pc
\sqcup \lab_e \sqsubseteq \lab$.\\
From IH, either $(\sigma(x) = \sigma'(x))$ or $\pc \sqcup \lab_e
\sqsubseteq \lab$. Thus, $(\sigma(x) = \sigma'(x)) \vee \pc \sqsubseteq \lab$.
\item while-t: IH1: If $\langle c, \sigma \rangle \Downarrow_{\pc
    \sqcup \lab_e} \sigma''$, then $\forall
x.\Gamma(\sigma''(x)) = \lab'' \implies (\sigma(x) = \sigma''(x)) \vee \pc
\sqsubseteq \lab''$.\\
IH2: If $\langle c, \sigma'' \rangle \Downarrow_{\pc \sqcup \lab_e} \sigma'$, then $\forall
x.\Gamma(\sigma'(x)) = \lab \implies (\sigma''(x) = \sigma'(x)) \vee \pc
\sqsubseteq \lab$.\\
From similar reasoning as in ``seq'', we have either $\sigma(x) =
\sigma'(x)$ or $\pc \sqcup \lab_e \sqsubseteq \lab$. Thus, $\sigma(x) =
\sigma'(x) \vee \pc \sqsubseteq \lab$. 
\item while-f:  $\sigma(x) = \sigma'(x)$. 
\end{itemize}
\end{proof}

\begin{myLemma}{\emph{Conf{i}nement Lemma}}
\label{conf}
If $\pc \not\sqsubseteq \lab$, $\langle c, \sigma \rangle \Downarrow_\pc
\sigma'$, then $\sigma \sim_\lab \sigma'$.
\end{myLemma}
\begin{proof}
Proof by induction on the derivation and case analysis on the last rule.
\begin{enumerate}
\item skip : $\sigma = \sigma'$.

\item assn: Let $x_i = v_i^{k_i}$ and $x_f =
  v_f^{k_f}$, s.t $k_i = \lab_i \vee k_i =   \lab_i \pl
   $. 
  \begin{itemize}
    \item  $\pc \sqsubseteq \lab_i$ : As $\pc \not\sqsubseteq \lab$, $\lab_i
      \not\sqsubseteq \lab$. By premises of
      assn, $k_f = \lab_f \vee k_f =  \lab_f\pl  $,
      where $\lab_f = \pc \sqcup \lab_e$. As $\pc \not\sqsubseteq \lab$,
      $\lab_f \not\sqsubseteq \lab$. 
      Thus, by definition \ref{def1}.2, \ref{def1}.3, \ref{def1}.4 or
      \ref{def1}.5, $x_i \sim_\lab x_f$. 

\item $\pc \not\sqsubseteq \lab_i$ :
  By premise, $k_f =   ((\pc \sqcup m) \sqcap \lab_i) \pl  $. Thus,
  $\lab_f \sqsubseteq \lab_i$ and by definition
  \ref{def1}.3 or~\ref{def1}.5 $x_i \sim_\lab x_f$.
\end{itemize}

\item seq : IH1: $\sigma \sim_\lab \sigma''$ and IH2: $\sigma''
  \sim_\lab \sigma'$. 
 T.S : $\sigma \sim_\lab \sigma'$.\\
  For all $x \in dom(\sigma)$,
 respective $x''  \in dom(\sigma'')$ and
 respective $x' \in dom(\sigma')$, $x \sim_\lab x''$ and $x''
 \sim_\lab x'$. \\
 To show: $x \sim_\lab x'$. \\
 Let $x = v_1^{k_1}, x'' = v_2^{k_2}, x' = v_3^{k_3}$, where $k_1 =
 \lab_1 \vee k_1 =   \lab_1\pl  $, $k_2 =
 \lab_2 \vee k_2 =   \lab_2\pl  $ and $k_3 =
 \lab_3 \vee k_3 =   \lab_3\pl  $.\\
Case-analysis on definition~\ref{def1} for IH1.
 \begin{itemize}

   \item $(k_1 = k_2) = \lab' \sqsubseteq \lab \wedge v_1 = v_2$ : By IH2
     and definition~\ref{def1}, we have
     \begin{enumerate}
       \item $(k_2 = k_3) = \lab' \sqsubseteq \lab \wedge v_2 = v_3$
         (case 1): Transitivity of equality, $(k_1 =
         k_3) = \lab' \sqsubseteq \lab \wedge v_1 = v_3$. Thus, $x
         \sim_\lab x'$.
        \item $k_2 = \lab'$ and $k_3  =   \lab_3 \pl   \wedge
          \lab_3 \sqsubseteq \lab' \sqsubseteq \lab$ (case 5): By definition~\ref{def1}.5 $x
          \sim_\lab x'$.
     \end{enumerate}

     \item $k_1 = \lab_1 \not\sqsubseteq \lab \wedge k_2 =
      \lab_2 \not\sqsubseteq \lab $: By IH2, either
      \begin{enumerate}
        \item $k_2 = \lab_2 \not\sqsubseteq \lab \wedge k_3 =
      \lab_3 \not\sqsubseteq \lab $. By definition~\ref{def1}.2, $x
         \sim_\lab x'$.
         \item $k_2 = \lab_2 \not\sqsubseteq \lab \wedge k_3 =
             \lab_3\pl  $: $\lab_1 \not\sqsubseteq
           \lab$. Thus, by definition~\ref{def1}.5, $x \sim_\lab x'$.
       \end{enumerate}

    \item $k_1 =   \lab_1  \pl   \wedge k_2 =
        \lab_2  \pl  $: By IH2, we have
      \begin{enumerate}
       \item $k_2 =   \lab_2  \pl   \wedge k_3 =
        \lab_3  \pl  $ (case 3): By
      definition~\ref{def1}.3, $x \sim_\lab x'$. 
      \item $k_2 =   \lab_2  \pl   \wedge k_3 =
      \lab_3 \wedge (\lab_3 \not\sqsubseteq \lab)$ (case 4):
      By definition~\ref{def1}.4, $x \sim_\lab x'$.
      \item $k_2 =   \lab_2  \pl   \wedge k_3 =
      \lab_3 \wedge (\lab_2 \sqsubseteq \lab_3)$ (case 4):  
      By corollary~\ref{cor1}, $\pc \sqsubseteq \lab_2$. As $\pc \not\sqsubseteq
      \lab$ and $\lab_2 \sqsubseteq \lab_3$, so $\lab_3
      \not\sqsubseteq \lab$. By definition~\ref{def1}.4, $x \sim_\lab
      x'$.
     .
     \end{enumerate}

   \item $k_1 =   \lab_1  \pl   \wedge k_2 =
      \lab_2~s.t.~(\lab_2 \not\sqsubseteq \lab)$ (case 4): Either
      \begin{itemize}
        \item $k_2 = \lab_2 \not\sqsubseteq \lab \wedge k_3 =
      \lab_3 \not\sqsubseteq \lab$: By definition~\ref{def1}.4, $x \sim_\lab
      x'$.
    \item $k_2 = \lab_2 \not\sqsubseteq \lab \wedge k_3 =
       \lab_3  \pl $: By definition~\ref{def1}.3, $x \sim_\lab
      x'$.
        \end{itemize}
    \item $k_1 =   \lab_1  \pl   \wedge k_2 =
      \lab_2~s.t.~ (\lab_1 \sqsubseteq \lab_2)$ (case 4): 
      \begin{itemize}
        \item $k_2 = k_3 = \lab_2$: By definition~\ref{def1}.4, $x \sim_\lab
      x'$.
    \item $k_2 = \lab_2 \not\sqsubseteq \lab \wedge k_3 =
      \lab_3 \not\sqsubseteq \lab$: By definition~\ref{def1}.4, $x \sim_\lab
      x'$.
    \item $k_2 = \lab_2 \not\sqsubseteq \lab \wedge k_3 =
        \lab_3  \pl $: By definition~\ref{def1}.3, $x \sim_\lab
      x'$.
        \end{itemize}

   \item $k_1 = \lab_1  \wedge k_2 =  
      \lab_2 \pl  ~s.t.~ (\lab_1 \not\sqsubseteq \lab)$: By
      IH2, we have 
      \begin{enumerate}
        \item $k_2 =   \lab_2  \pl   \wedge k_3 =
        \lab_3  \pl  $ (case 3): By definition~\ref{def1}.5, $x
         \sim_\lab x'$.
        \item  $k_2 =   \lab_2  \pl   \wedge k_3 =
      \lab_3~s.t.~(\lab_3 \not\sqsubseteq \lab)$ (case 4): By
      definition~\ref{def1}.2, $x \sim_\lab x'$.
      \item  $k_2 =   \lab_2  \pl   \wedge k_3 =
      \lab_3~s.t.~(\lab_2 \sqsubseteq \lab_3)$ (case 4): By corollary~\ref{cor1},
      $\pc \sqsubseteq \lab_2$. As $\pc \not\sqsubseteq 
      \lab$ and $\lab_2 \sqsubseteq \lab_3$, so $\lab_3
      \not\sqsubseteq \lab$. By definition~\ref{def1}.2, $x \sim_\lab
      x'$.
     \end{enumerate}

      \item $k_1 = \lab_1  \wedge k_2 =  
      \lab_2 \pl  ~s.t.~ (\lab_2 \sqsubseteq \lab_1)$: Also,
      $(\lab_2 \sqsubseteq \lab_1 \sqsubseteq \lab)$. By
      IH2, we have 
      \begin{enumerate}
        \item $k_2 =   \lab_2  \pl   \wedge k_3 =
        \lab_3  \pl  $ (case 3): As $\lab_2 \sqsubseteq
      \lab$ and $\pc \not\sqsubseteq \lab$, $\pc \not\sqsubseteq
      \lab_2$.  By lemma~\ref{sup1}, $\lab_3 \sqsubseteq \lab_2 $. Thus,
      $\lab_3\sqsubseteq \lab_2 \sqsubseteq \lab_1$. By definition~\ref{def1}.5, $x
         \sim_\lab x'$.
        \item  $k_2 =   \lab_2  \pl   \wedge k_3 =
      \lab_3$ (case 4): As $\lab_2 \sqsubseteq
      \lab$ and $\pc \not\sqsubseteq \lab$, $\pc \not\sqsubseteq
      \lab_2$. But, by corollary~\ref{cor1}, $\pc \sqsubseteq \lab_2$. By
      contradiction, this case      does not hold.
     \end{enumerate}
 \end{itemize}
 
\item if-else : IH : $k = \lab'$. If $(pc \sqcup \lab')
   \not\sqsubseteq \lab $, then $\sigma \sim_\lab \sigma'$. \\
  As $pc \not\sqsubseteq \lab$, $pc \sqcup \lab' \not\sqsubseteq 
  \lab$. Thus, by IH, $\sigma \sim_\lab \sigma'$. 
 
 \item while-t: IH1 : $k = \lab'$. If $(pc \sqcup \lab')
   \not\sqsubseteq \lab $, then $\sigma \sim_\lab \sigma'$. \\
   As $pc \not\sqsubseteq \lab$, $pc \sqcup \lab' \not\sqsubseteq
  \lab$. Thus, by IH1, $\sigma \sim_\lab \sigma''$. \\
   IH2 : $k = \lab'$. If $(pc \sqcup \lab')
   \not\sqsubseteq \lab $, then $\sigma' \sim_\lab \sigma''$. \\
As $pc \not\sqsubseteq \lab$, $pc \sqcup \lab' \not\sqsubseteq
  \lab$. Thus, by IH, $\sigma'' \sim_\lab \sigma'$. \\
  Therefore, $\sigma \sim_\lab \sigma''$ and $\sigma'' \sim_\lab
  \sigma'$. \\ 
 (Reasoning similar to seq.)
 
\item while-f : $\sigma = \sigma'$
\end{enumerate}
\end{proof}

\begin{myThm}{\emph{Termination-insensitive non-interference}}\\
\label{ni}
If
$~\sigma_1 \sim_\lab \sigma_2$, \\
$\langle c, \sigma_1  \rangle \Downarrow_\pc \sigma_1' $,\\
$\langle c, \sigma_2   \rangle\Downarrow_\pc \sigma_2' $, \\
then \\
$\sigma_1' \sim_\lab \sigma_2'$.
\end{myThm}
\begin{proof}
By induction on the derivation and case analysis on the last step
\begin{enumerate}
 \item skip: $\sigma_1' = \sigma_1 \sim_\lab \sigma_2 = \sigma_2'$

 \item assn($x := e$): As $\sigma_1 \sim_\lab \sigma_2$, 
   $\forall x. \sigma_1(x) \sim_\lab \sigma_2(x)$. Let $\sigma_1(x) =
   v_1^{k_1}$, $\sigma_2(x) = v_2^{k_2}$ and \\ $\sigma_1'(x) =
   v_1'^{k_1'}$, $\sigma_2'(x) = v_2'^{k_2'}$ \\ s. t. $k_i = \lab_i
   \vee k_i =   \lab_i \pl  $ and $k'_i = \lab'_i
   \vee k_i =   \lab'_i \pl  $ for $i=1,2$. \\
   Let
   $\langle  e_1, \sigma_1 \rangle \Downarrow
   w_1^{k^{e}_1} \wedge \langle  e_2, \sigma_2 \rangle \Downarrow
   w_2^{k^{e}_2}$ \\ s. t. $k^{e}_i =
   \lab^{e}_i \vee k^{e}_i =   \lab^{e}_i \pl  $ for $i
   =1,2$.
      For low-equivalence of $e_1$ and $e_2$, the following cases
       arise:
       \begin{enumerate}
         
          \item $k^{e}_i = \lab^{e}_i,~s.t.~ (\lab^e_1 = \lab^e_2) = \lab^e \sqsubseteq \lab \wedge
          w_1 = w_2$:
          \begin{enumerate}
            \item $\pc \not\sqsubseteq \lab_1 \wedge \pc
              \not\sqsubseteq \lab_2$: By premise of assn rules, $k_i'
              =   ((\pc \sqcup \lab^e) \sqcap \lab_i) \pl $. By
              definition~\ref{def1}.3, $\sigma_1' \sim_\lab
              \sigma_2'$.
            \item $\pc \not\sqsubseteq \lab_1 \wedge \pc
              \sqsubseteq \lab_2$: $k_1' =   ((\pc \sqcup \lab^e) \sqcap \lab_1) \pl
               $ and $k_2' = \pc \sqcup \lab^e$. As $\lab_1'
              \sqsubseteq \lab_2'$, by
              definition~\ref{def1}.4, $\sigma_1' \sim_\lab
              \sigma_2'$.
            \item $\pc \sqsubseteq \lab_1 \wedge \pc
              \not \sqsubseteq \lab_2$: $k_2' =   ((\pc \sqcup \lab^e) \sqcap \lab_2) \pl
               $ and $k_1' = \pc \sqcup \lab^e$. As $\lab_2'
              \sqsubseteq \lab_1'$, by
              definition~\ref{def1}.5, $\sigma_1' \sim_\lab
              \sigma_2'$.
            \item $\pc \sqsubseteq \lab_1 \wedge \pc
              \sqsubseteq \lab_2$: $k_1' = \pc \sqcup \lab^e$ and
              $k_2' = \pc \sqcup \lab^e$. If $\pc \sqsubseteq \lab$
              and $\lab^e \sqsubseteq \lab$ and $w_1 = w_2$, by 
              definition~\ref{def1}.1, $\sigma_1' \sim_\lab
              \sigma_2'$. If $\pc \not\sqsubseteq \lab$, $\pc \sqcup
              \lab^e \not\sqsubseteq \lab$. By
              definition~\ref{def1}.2, $\sigma_1' \sim_\lab
              \sigma_2'$.
          \end{enumerate}
       
        \item $\lab^e_1 \not\sqsubseteq \lab \wedge
          \lab^e_2 \not\sqsubseteq \lab$: 
          From premise of assignment rules, 
          $k_1'= \pc \sqcup \lab^e_1 \vee
          k_1'=   (\pc \sqcup \lab^e_1) \pl
            \vee  k_1' =   ((\pc \sqcup \lab^e_1) \sqcap
          \lab_1) \pl  $.
          Similarly, 
          $k_2' = \pc \sqcup \lab^{e}_2 \vee
          k_2'=   (\pc \sqcup \lab^{e}_2) \pl
            \vee  k_2'=   ((\pc \sqcup \lab^e_2) \sqcap
          \lab_2) \pl  $.
          Since $\lab^{e}_1  \not\sqsubseteq \lab$ and $\lab^{e}_2
          \not\sqsubseteq \lab$, $\pc \sqcup \lab^{e}_1  \not\sqsubseteq \lab$
          and $\pc \sqcup \lab^{e}_1  \not\sqsubseteq \lab$. Therefore, from
          Definition~\ref{def1}.2, \ref{def1}.3, \ref{def1}.4 or
          \ref{def1}.5 $\sigma_1' \sim_\lab \sigma_2'$.        
 
       \item $k^{e}_i =   \lab^{e}_i \pl  $: By premise of assn rules, $k_i'
              =   ((\pc \sqcup \lab^e_i) \sqcap \lab_i) \pl  $ or $k_i'
              =   (\pc \sqcup \lab^e_i) \pl  $. By
              definition~\ref{def1}.3, $\sigma_1' \sim_\lab
              \sigma_2'$. 

        \item $k^{e}_1 =   \lab^{e}_1 \pl   \wedge
          k^{e}_2 = \lab^e_2$: 
         \begin{enumerate}
            \item $\pc \not\sqsubseteq \lab_1 \wedge \pc
              \not\sqsubseteq \lab_2$: By premise of assn rules, $k_i'
              =   ((\pc \sqcup \lab^e_i) \sqcap \lab_i) \pl  $. By
              definition~\ref{def1}.3, $\sigma_1' \sim_\lab
              \sigma_2'$.
            \item $\pc \not\sqsubseteq \lab_1 \wedge \pc
              \sqsubseteq \lab_2$: $k_1' =   ((\pc \sqcup \lab^e_1) \sqcap \lab_1) \pl
               $ and $k_2' = \pc \sqcup \lab^e_2$. From
               definition~\ref{def1}.4, $\lab^e_1 \sqsubseteq \lab^e_2$,
               so $(\pc \sqcup \lab^e_i)\sqcap
              \lab_1 \sqsubseteq \pc \sqcup \lab^e_2$. By
              definition~\ref{def1}.4, $\sigma_1' \sim_\lab
              \sigma_2'$.
            \item $\pc \sqsubseteq \lab_1 \wedge \pc
              \not \sqsubseteq \lab_2$: $k_2' =   ((\pc \sqcup
              \lab^e_2) \sqcap \lab_2) \pl 
               $ and $k_1' =   (\pc \sqcup \lab^e_1)
              \pl $. By
              definition~\ref{def1}.3, $\sigma_1' \sim_\lab
              \sigma_2'$.
            \item $\pc \sqsubseteq \lab_1 \wedge \pc
              \sqsubseteq \lab_2$: $k_1' =   (\pc \sqcup
              \lab^e_1) \pl  $ and $k_2' = \pc \sqcup
              \lab^e_2$. If $\lab^e_2 \not\sqsubseteq \lab$, so $\pc
              \sqcup \lab^e_2 \not\sqsubseteq \lab$. Else if $\lab^e_1
              \sqsubseteq \lab^e_2$, then $\pc \sqcup \lab^e_1
              \sqsubseteq \pc \sqcup \lab^e_2$. By
              definition~\ref{def1}.4, $\sigma_1' \sim_\lab
              \sigma_2'$.
          \end{enumerate}

       \item $k^{e}_1 = \lab^{e}_1 \wedge
          k^{e}_2 =   \lab^e_2 \pl  $: 
         \begin{enumerate}
            \item $\pc \not\sqsubseteq \lab_1 \wedge \pc
              \not\sqsubseteq \lab_2$: By premise of assn rules, $k_i'
              =   ((\pc \sqcup \lab^e_i) \sqcap \lab_i) \pl  $. By
              definition~\ref{def1}.3, $\sigma_1' \sim_\lab
              \sigma_2'$.
            \item $\pc \not\sqsubseteq \lab_1 \wedge \pc
              \sqsubseteq \lab_2$: $k_1' =   ((\pc \sqcup
              \lab^e_1)\sqcap \lab_1) \pl 
               $ and $k_2' =   (\pc \sqcup \lab^e_2)
              \pl $. By
              definition~\ref{def1}.3, $\sigma_1' \sim_\lab
              \sigma_2'$.
            \item $\pc \sqsubseteq \lab_1 \wedge \pc
              \not \sqsubseteq \lab_2$: $k_1' = \pc \sqcup \lab^e_1$
              and $k_2' =   ((\pc \sqcup \lab^e_2) \sqcap \lab_2) \pl
               $. $(\pc \sqcup \lab^e_2) \sqcap
              \lab_2 \sqsubseteq \pc \sqcup \lab^e_1$. By
              definition~\ref{def1}.5, $\sigma_1' \sim_\lab
              \sigma_2'$.
            \item $\pc \sqsubseteq \lab_1 \wedge \pc
              \sqsubseteq \lab_2$: $k_1' = (\pc \sqcup
              \lab^e_1) \pl $ and $k_2' =  \pc \sqcup
              \lab^e_2  $. If $\lab^e_1 \not\sqsubseteq \lab$, so $\pc
              \sqcup \lab^e_1 \not\sqsubseteq \lab$. Else if $\lab^e_2
              \sqsubseteq \lab^e_1$, then $\pc \sqcup \lab^e_2
              \sqsubseteq \pc \sqcup \lab^e_1$. By
              definition~\ref{def1}.5, $\sigma_1' \sim_\lab
              \sigma_2'$.
          \end{enumerate}
       \end{enumerate}

 \item seq: IH1: If $\sigma_1 \sim_\lab \sigma_2$ then $\sigma_1''
   \sim_\lab \sigma_2''$ \\
   IH2: If $\sigma_1'' \sim_\lab \sigma_2''$ then $\sigma_1'
   \sim_\lab \sigma_2'$ \\
   Since $\sigma_1 \sim_\lab \sigma_2$, therefore, from IH1 and IH2 
   $\sigma_1' \sim_\lab \sigma_2'$.

 \item if-else: IH: If $\sigma_1 \sim_\lab \sigma_2$, $ \langle c,
   \sigma_1  \rangle \Downarrow_{\pc \sqcup \lab^e_1} \sigma_1'$, $  \langle c,
   \sigma_2  \rangle \Downarrow_{\pc \sqcup \lab^e_2} \sigma_2'$ and
   $\pc \sqcup \lab^e_1 = \pc \sqcup \lab^e_2$ then $\sigma_1'
   \sim_\lab \sigma_2'$. 
   \begin{itemize}
     \item If $\lab^e_1 \sqsubseteq \lab$, $\lab^e_1 = \lab^e_2$ and $n_1 =
   n_2$. By IH, $\sigma_1' \sim_\lab \sigma_2'$.
   \item If $\lab^e_1 \not\sqsubseteq \lab$, then $\lab^e_2
     \not\sqsubseteq \lab$, $\pc \sqcup \lab^e_i \not\sqsubseteq
   \lab$ for $i =1,2$. By Lemma~\ref{conf}, $\sigma_1 \sim_\lab \sigma_1'$ and $\sigma_2
   \sim_\lab \sigma_2'$.
   T.S. $\sigma_1' \sim_\lab \sigma_2'$, i.e., $(\forall x. \sigma_1'(x) \sim_\lab \sigma_2'(x))$\\
   Case analysis on the definition of low-equivalence of values, $x$, in $\sigma_1$ and
$\sigma_2$. Let $\sigma_1(x) = v_1^{k_1}$ and $\sigma_2(x) =
v_2^{k_2}$ and $\sigma_1'(x) = v_1'^{k_1'}$ and $\sigma_2'(x) =
v_2'^{k_2'}$
\begin{enumerate}
\item $(k_1 = k_2) = \lab' \sqsubseteq \lab ~\wedge~v_1 = v_2 = v$:
 \begin{itemize}
\item   If $k_1' = \lab_1' \wedge k_2' = \lab_2'$, then as $\sigma_1
  \sim_\lab \sigma_1'$ and $\sigma_2 \sim_\lab \sigma_2'$, by
  definition~\ref{def1}.1, $\lab' = \lab_1' \wedge v = v_1'$ and
  $\lab' = \lab_2' \wedge v = v_2'$. Thus, $\lab_1' = \lab_2' \wedge
  v_1' = v_2'$, so $\sigma_1'(x)  \sim_\lab \sigma_2'(x)$.
\item  If $k_1' =  \lab_1'\pl   ~\wedge~k_2' = \lab_2'$, then as $\sigma_1
  \sim_\lab \sigma_1'$ and $\sigma_2 \sim_\lab \sigma_2'$, by
  definition~\ref{def1}.5 $\lab_1' \sqsubseteq \lab_1 = \lab'$ and by
  definition~\ref{def1}.1 
  $k_2' = \lab_2' =\lab_2 = \lab'$. So, $\lab_1' \sqsubseteq \lab_2'$.
  By definition~\ref{def1}.4, $\sigma_1'(x)
  \sim_\lab \sigma_2'(x)$.
\item If $k_1' = \lab_1' ~\wedge~k_2' =  \lab_2'\pl  $, then as $\sigma_1
  \sim_\lab \sigma_1'$ and $\sigma_2 \sim_\lab \sigma_2'$,by
  definition~\ref{def1}.1 $k_1' = \lab_1' =\lab_1 = \lab'$ and by
  definition~\ref{def1}.5 $\lab_2' \sqsubseteq \lab_2 = \lab'$. So,
  $\lab_2' \sqsubseteq \lab_1'$. 
  By definition~\ref{def1}.5, $\sigma_1'(x)
  \sim_\lab \sigma_2'(x)$.
\item If $k_1' =  \lab_1'\pl   ~\wedge~k_2' =
 \lab_2'\pl  $, then by definition~\ref{def1}.3,
$\sigma_1'(x)   \sim_\lab \sigma_2'(x)$.
 \end{itemize}

\item $(k_1 = \lab_1 \not\sqsubseteq \lab)\wedge (k_2 = \lab_2
  \not\sqsubseteq \lab)$:
 \begin{itemize}
\item  If $k_1' = \lab_1' \wedge k_2' = \lab_2'$, then as $\sigma_1
  \sim_\lab \sigma_1'$ and $\sigma_2 \sim_\lab \sigma_2'$, by
  definition~\ref{def1}.2, $(k_1' = \lab_1' \not\sqsubseteq
  \lab)\wedge (k_2' = \lab_2' \not\sqsubseteq \lab)$. So,
  $\sigma_1'(x)  \sim_\lab \sigma_2'(x)$.
 \item If $k_1' =  \lab_1'\pl   ~\wedge~k_2' = \lab_2'$, then as $\sigma_1
  \sim_\lab \sigma_1'$ and $\sigma_2 \sim_\lab \sigma_2'$, by
  definition~\ref{def1}.2  $k_2' = \lab_2' \not\sqsubseteq \lab$.
  By definition~\ref{def1}.4, $\sigma_1'(x)
  \sim_\lab \sigma_2'(x)$.
\item If $k_1' = \lab_1' ~\wedge~k_2' =  \lab_2'\pl  $, then as $\sigma_1
  \sim_\lab \sigma_1'$ and $\sigma_2 \sim_\lab \sigma_2'$,by
  definition~\ref{def1}.2 $k_1' = \lab_1' \not\sqsubseteq \lab$. 
  By definition~\ref{def1}.5, $\sigma_1'(x)
  \sim_\lab \sigma_2'(x)$.
If $k_1' =  \lab_1'\pl   ~\wedge~k_2' =
 \lab_2'\pl  $, then by definition~\ref{def1}.3,
$\sigma_1'(x)   \sim_\lab \sigma_2'(x)$.
 \end{itemize}

\item $(k_1 =   \lab_1\pl   ~\wedge~k_2 =  
  \lab_2\pl  )$ :
 \begin{itemize}
\item If $k_1' =  \lab_1'\pl   ~\wedge~k_2' =
 \lab_2'\pl  $, by definition~\ref{def1}.3,
$\sigma_1'(x)   \sim_\lab \sigma_2'(x)$.
\item If $k_1' = \lab_1' ~\wedge~k_2' =  \lab_2'\pl  $, then as $\sigma_1
  \sim_\lab \sigma_1'$ and $\sigma_2 \sim_\lab \sigma_2'$,by 
corollary~\ref{cor2}, $\pc \sqcup \lab^e_1 \sqsubseteq \lab_1'$. As $\pc \sqcup \lab^e_1 \not\sqsubseteq \lab$ and
  by  definition~\ref{def1}.2,
  $\lab_1' \not\sqsubseteq \lab$.
  By definition~\ref{def1}.5, $\sigma_1'(x)
  \sim_\lab \sigma_2'(x)$.
\item  If $k_1' =  \lab_1'\pl   ~\wedge~k_2' = \lab_2'$, then as $\sigma_1
  \sim_\lab \sigma_1'$ and $\sigma_2 \sim_\lab \sigma_2'$, by 
corollary~\ref{cor2}, $\pc \sqcup \lab^e_2 \sqsubseteq \lab_2'$. As $\pc \sqcup
  \lab^e_2 \not\sqsubseteq \lab$ and
  by definition~\ref{def1}.2,
  $\lab_2' \not\sqsubseteq \lab$.
  By definition~\ref{def1}.4, $\sigma_1'(x)
  \sim_\lab \sigma_2'(x)$.
\item If $k_1' = \lab_1' \wedge k_2' = \lab_2'$, then as $\sigma_1
  \sim_\lab \sigma_1'$ and $\sigma_2 \sim_\lab \sigma_2'$, by 
  corollary~\ref{cor2}, $\pc \sqcup \lab^e_1 \sqsubseteq \lab_1'$ and $\pc \sqcup \lab^e_2 \sqsubseteq \lab_2'$. As
  $\pc \sqcup \lab^e_i \not\sqsubseteq \lab$ and by definition~\ref{def1}.2, $\lab_1'
  \not\sqsubseteq \lab$ and 
  $\lab_2' \not\sqsubseteq \lab$.
  By definition~\ref{def1}.2, $\sigma_1'(x)
  \sim_\lab \sigma_2'(x)$.
 \end{itemize}

\item $(k_1 =   \lab_1 \pl   ~\wedge~k_2 = 
  \lab_2)$: 
\begin{itemize}
\item $\lab_2 \not\sqsubseteq \lab$ : 
 \begin{itemize}
\item  If $k_1' =  \lab_1'\pl   ~\wedge~k_2' =
 \lab_2'\pl  $, by definition~\ref{def1}.3,
$\sigma_1'(x)   \sim_\lab \sigma_2'(x)$.
\item If $k_1' = \lab_1' ~\wedge~k_2' =  \lab_2'\pl  $, then as $\sigma_1
  \sim_\lab \sigma_1'$ and $\sigma_2 \sim_\lab \sigma_2'$,by
  corollary~\ref{cor2},  $\pc \sqcup \lab^e_1\sqsubseteq \lab_1'$. As $\pc \sqcup \lab^e_1 \not\sqsubseteq \lab$ and by
  definition~\ref{def1}.2, $\lab_1' \not\sqsubseteq \lab$.  By
  definition~\ref{def1}.5, $\sigma_1'(x) 
  \sim_\lab \sigma_2'(x)$.
\item  If $k_1' =  \lab_1'\pl   ~\wedge~k_2' = \lab_2'$, then as $\sigma_1
  \sim_\lab \sigma_1'$ and $\sigma_2 \sim_\lab \sigma_2'$, by
  definition~\ref{def1}.2, $\lab_2' \not\sqsubseteq \lab$.
  By definition~\ref{def1}.4, $\sigma_1'(x)
  \sim_\lab \sigma_2'(x)$.
\item If $k_1' = \lab_1' \wedge k_2' = \lab_2'$, then as $\sigma_1
  \sim_\lab \sigma_1'$ and $\sigma_2 \sim_\lab \sigma_2'$, by
  corollary~\ref{cor2},  $\pc \sqcup \lab^e_1\sqsubseteq \lab_1'$. As $\pc \sqcup \lab^e_1 \not\sqsubseteq \lab$ and by
  definition~\ref{def1}.2, $\lab_1' \not\sqsubseteq \lab$. By
  definition~\ref{def1}.2, $\lab_2' \not\sqsubseteq \lab$.
  By definition~\ref{def1}.2, $\sigma_1'(x)
  \sim_\lab \sigma_2'(x)$.
 \end{itemize}
\item $\lab_1 \sqsubseteq \lab_2 \sqsubseteq \lab$ : 
 \begin{itemize}
\item  If $k_1' =  \lab_1'\pl   ~\wedge~k_2' =
 \lab_2'\pl  $, by definition~\ref{def1}.3,
$\sigma_1'(x)   \sim_\lab \sigma_2'(x)$.
\item If $k_1' = \lab_1' ~\wedge~k_2' =  \lab_2'\pl  $, then as $\sigma_1
  \sim_\lab \sigma_1'$ and $\sigma_2 \sim_\lab \sigma_2'$, by
  corollary~\ref{cor2},  $\pc \sqcup \lab^e_1 \sqsubseteq \lab_1'$. As $\pc \sqcup
  \lab^e_1 \not\sqsubseteq \lab$, and by definition ~\ref{def1}.2,
  $\lab_1' \not\sqsubseteq \lab$. By definition~\ref{def1}.5,
  $\sigma_1'(x) \sim_\lab \sigma_2'(x)$.
\item  If $k_1' =  \lab_1'\pl   ~\wedge~k_2' = \lab_2'$, then as $\sigma_1
  \sim_\lab \sigma_1'$ and $\sigma_2 \sim_\lab \sigma_2'$, $\lab_1' \sqsubseteq
  (\pc \sqcup \lab^e_1)\sqcap \lab_1$ as $\pc \sqcup \lab^e_1\not\sqsubseteq \lab_1$ and $\lab_2' =
  \lab_2$ by corollary~\ref{cor1} and definition~\ref{def1}.1. Thus,
  $\lab_1'\sqsubseteq \lab_2'$.  By definition~\ref{def1}.4, $\sigma_1'(x)
  \sim_\lab \sigma_2'(x)$.
\item If $k_1' = \lab_1' \wedge k_2' = \lab_2'$, then as $\sigma_1
  \sim_\lab \sigma_1'$ and $\sigma_2 \sim_\lab \sigma_2'$, by 
  corollary~\ref{cor1}, $\pc \sqcup \lab^e_1\sqsubseteq \lab_1$. As $\pc \sqcup
  \lab^e_1 \not\sqsubseteq \lab$, by 
  contradiction the case does not hold.
 \end{itemize}
\end{itemize}

\item $(k_1 = \lab_1 ~\wedge~k_2 = 
    \lab_2 \pl  )$:
\begin{itemize}
 \item $\lab_1 \not\sqsubseteq \lab$ : 
\begin{itemize}
\item  If $k_1' =  \lab_1'\pl   ~\wedge~k_2' =
 \lab_2'\pl  $, by definition~\ref{def1}.3,
$\sigma_1'(x)   \sim_\lab \sigma_2'(x)$.
\item If $k_1' =   \lab_1'\pl  ~\wedge~k_2' = \lab_2' $, then as $\sigma_1
  \sim_\lab \sigma_1'$ and $\sigma_2 \sim_\lab \sigma_2'$,by 
corollary~\ref{cor2}, $\pc \sqcup \lab^e_2 \sqsubseteq \lab_2'$. As $\pc \sqcup \lab^e_2 \not\sqsubseteq \lab$ and by
  definition~\ref{def1}.2, $\lab_2' \not\sqsubseteq \lab$.  By
  definition~\ref{def1}.5, $\sigma_1'(x) 
  \sim_\lab \sigma_2'(x)$.
\item  If $k_1' = \lab_1' ~\wedge~k_2' =  
 \lab_2' \pl  $, then as $\sigma_1
  \sim_\lab \sigma_1'$ and $\sigma_2 \sim_\lab \sigma_2'$, by
  definition~\ref{def1}.2, $\lab_1' \not\sqsubseteq \lab$.
  By definition~\ref{def1}.4, $\sigma_1'(x)
  \sim_\lab \sigma_2'(x)$.
\item If $k_1' = \lab_1' \wedge k_2' = \lab_2'$, then as $\sigma_1
  \sim_\lab \sigma_1'$ and $\sigma_2 \sim_\lab \sigma_2'$, by 
  corollary~\ref{cor2}, $\pc \sqcup \lab^e_2\sqsubseteq \lab_2'$. As $\pc \sqcup \lab^e_2\not\sqsubseteq \lab$ and by
  definition~\ref{def1}.2, $\lab_2' \not\sqsubseteq \lab$. By
  definition~\ref{def1}.2, $\lab_1' \not\sqsubseteq \lab$.
  By definition~\ref{def1}.2, $\sigma_1'(x)
  \sim_\lab \sigma_2'(x)$.
\end{itemize}
\item $\lab_2 \sqsubseteq \lab_1$ :
\begin{itemize}
\item If $k_1' =  \lab_1'\pl   ~\wedge~k_2' =
 \lab_2'\pl  $, by definition~\ref{def1}.3,
$\sigma_1'(x)   \sim_\lab \sigma_2'(x)$.
\item If $k_1' = \lab_1' ~\wedge~k_2' =  \lab_2'\pl  $, then as $\sigma_1
  \sim_\lab \sigma_1'$ and $\sigma_2 \sim_\lab \sigma_2'$, $\lab_2' \sqsubseteq
  (\pc \sqcup \lab^e_2)\sqcap \lab_2$ as $\pc \sqcup \lab^e_2\not\sqsubseteq \lab_2$ and $\lab_1' =
  \lab_1$ by corollary~\ref{cor1} and definition~\ref{def1}.1. Thus,
  $\lab_2'\sqsubseteq \lab_1'$.  By definition~\ref{def1}.5, $\sigma_1'(x)
  \sim_\lab \sigma_2'(x)$.
 \item If $k_1' =  \lab_1'\pl   ~\wedge~k_2' = \lab_2'$, then as $\sigma_1
  \sim_\lab \sigma_1'$ and $\sigma_2 \sim_\lab \sigma_2'$, by
  corollary~\ref{cor2},  $\pc \sqcup \lab^e_2\sqsubseteq \lab_2'$. As $\pc \sqcup \lab^e_2\not\sqsubseteq \lab$, and by
  definition~\ref{def1}.2, $\lab_2' \not\sqsubseteq \lab$.  By
  definition~\ref{def1}.4, $\sigma_1'(x) 
  \sim_\lab \sigma_2'(x)$.
\item If $k_1' = \lab_1' \wedge k_2' = \lab_2'$, then as $\sigma_1
  \sim_\lab \sigma_1'$ and $\sigma_2 \sim_\lab \sigma_2'$, by
  corollary~\ref{cor1},  $\pc \sqcup \lab^e_2\sqsubseteq \lab_2$. As $\pc \sqcup \lab^e_2\not\sqsubseteq \lab$, by
  contradiction the case does not hold.
  \end{itemize}
 \end{itemize}
\end{enumerate}
\end{itemize}

 \item while-t:  IH1: If $\sigma_1 \sim_\lab \sigma_2$, $ \langle c,
   \sigma_1   \rangle\Downarrow_{\pc \sqcup \lab^e_1} \sigma_1''$, $ \langle c,
   \sigma_2 \rangle \Downarrow_{\pc \sqcup \lab^e_2} \sigma_2''$ and
   $\pc \sqcup \lab^e_1 = \pc \sqcup \lab^e_2$ then $\sigma_1''
   \sim_\lab \sigma_2''$.\\
   IH2: If $\sigma_1'' \sim_\lab \sigma_2''$, $ \langle c,
   \sigma_1'' \rangle \Downarrow_{\pc \sqcup \lab^e_1} \sigma_1'$, $\langle c,
   \sigma_2'' \rangle \Downarrow_{\pc \sqcup \lab^e_2} \sigma_2'$ and
   $\pc \sqcup \lab^e_1 = \pc \sqcup \lab^e_2$ then $\sigma_1'
   \sim_\lab \sigma_2'$.
   \begin{itemize}
     \item If $\lab_1^e \sqsubseteq \lab$, $\lab_1^e = \lab^e_2$ and
       $n_1=n_2$. By IH1 and IH2, $\sigma_1' \sim_\lab \sigma_2'$.
      \item If $\lab^e_1 \not\sqsubseteq \lab$, then $\lab^e_2
        \not\sqsubseteq \lab$, $\pc \sqcup \lab^e_i \not\sqsubseteq
        \lab$ for $i = 1,2$. By Lemma~\ref{conf}, $\sigma_1 \sim_\lab
        \sigma_1''$ and $\sigma_2 \sim_\lab \sigma_2''$. \\
        T.S. $\sigma_1'' \sim_\lab \sigma_2''$: By similar reasoning
        as if-else.\\
        As $\sigma_1'' \sim_\lab \sigma_2''$, and by Lemma~\ref{conf}, $\sigma_1'' \sim_\lab
        \sigma_1'$ and $\sigma_2'' \sim_\lab \sigma_2'$. \\
        T.S. $\sigma_1' \sim_\lab \sigma_2'$: By similar reasoning
        as if-else.
    \end{itemize}

  \item while-f: $\sigma_1' = \sigma_1 \sim_\lab \sigma_2 = \sigma_2'$
\end{enumerate}

\end{proof}